\documentclass[10pt,journal,compsoc]{IEEEtran}
\ifCLASSOPTIONcompsoc
  \usepackage[nocompress]{cite}
\else
  \usepackage{cite}
\fi
\ifCLASSINFOpdf
\else
\fi
\usepackage{bm}
\usepackage{theorem}
\usepackage{graphics}
\usepackage{graphicx}

\usepackage{amsmath}
\usepackage{latexsym}
\usepackage{amssymb}

\makeatletter
\def\eqnarray{\stepcounter{equation}\let\@currentlabel=\theequation
\global\@eqnswtrue
\global\@eqcnt\z@\tabskip\@centering\let\\=\@eqncr
$$\halign to \displaywidth\bgroup\@eqnsel\hskip\@centering
  $\displaystyle\tabskip\z@{##}$&\global\@eqcnt\@ne 
  \hfil$\;{##}\;$\hfil
  &\global\@eqcnt\tw@ $\displaystyle\tabskip\z@{##}$\hfil 
   \tabskip\@centering&\llap{##}\tabskip\z@\cr}
\makeatother
\theorembodyfont{\itshape}
\newtheorem{thm}{Theorem}[section]
\newtheorem{lem}{Lemma}[section]
\newtheorem{prop}{Proposition}[section]
\newtheorem{coro}{Corollary}[section]

\theorembodyfont{\rmfamily}
\newtheorem{assumpt}{Assumption}
{\bf}{\rm}

\newtheorem{rem}{Remark}

\newenvironment{proof}{\noindent{\it Proof.~~}}{\qed\medskip}

\def\vc#1{\mbox{\boldmath $#1$}}

\newcommand{\qed}{\hspace*{\fill}$\Box$}
\newcommand{\dm}{\displaystyle}

\newcommand{\rmd}{\mathrm{d}}

\newcommand{\calS}{\mathcal{S}}
\newcommand{\calN}{\mathcal{N}}

\newcommand{\calP}{\mathcal{P}}

\newcommand{\bbN}{\mathbb{N}}
\newcommand{\bbR}{\mathbb{R}}

\newcommand{\bbZ}{\mathbb{Z}}
\newcommand{\bbE}{\mathbb{E}}

\newcommand{\Var}{\mathrm{Var}}
\newcommand{\Cov}{\mathrm{Cov}}

\newcommand{\dB}{\mathrm{dB}}

\newcommand{\rsp}{\mathrm{sp}}
\newcommand{\rtm}{\mathrm{tm}}

\newcommand{\eps}{\varepsilon}

\newcommand{\dd}[1]{\if#11 1\!\!1 
\else {\if#1C I\!\!\!C
\else {\if#1G I\!\!\!G 
\else {\if#1J J\!\!\!J 
\else {\if#1S S\!\!\!S
\else {\if#1Z Z\!\!\!Z
\else {\if#1Q O\!\!\!\!Q
\else I\!\!#1
\fi} 
\fi}
\fi}
\fi} 
\fi} 
\fi} 
\fi} 

\begin{document}
%
\title{Spatio-Temporal Correlation of Interference in MANET Under Spatially Correlated Shadowing Environment}
%
%
%
%

\author{Tatsuaki~Kimura,~\IEEEmembership{Member,~IEEE,}
        and~Hiroshi~Saito,~\IEEEmembership{Fellow,~IEEE}
\IEEEcompsocitemizethanks{
\IEEEcompsocthanksitem T. Kimura was with NTT Network Technology Laboratories, NTT Corporation,
Tokyo, Japan, 
and is currently with Department of Information and Communications Technology, 
Graduate School of Engineering, Osaka University,
Osaka, Japan.\protect\\ 
E-mail: kimura@comm.eng.osaka-u.ac.jp\protect\\
\IEEEcompsocthanksitem H. Saito was with NTT Network Technology Laboratories, NTT Corporation,
Tokyo, Japan, 
and is currently with Mathematics and Informatics Center, the University of Tokyo,
Tokyo, Japan.\protect\\ 
E-mail: saito@mi.u-tokyo.ac.jp
}

\thanks{
\copyright~2019 IEEE.  Personal use of this material is permitted.  Permission from IEEE must be obtained for all other uses, in any current or future media, including reprinting/republishing this material for advertising or promotional purposes, creating new collective works, for resale or redistribution to servers or lists, or reuse of any copyrighted component of this work in other works.
}
}

%
%

\markboth{}
{Kimura \MakeLowercase{\textit{et al.}}: Spatio-temporal correlation of interference}
\IEEEtitleabstractindextext{%
\begin{abstract}
Correlation of interference affects spatio-temporal aspects of various wireless mobile 
systems, such as retransmission, multiple antennas and cooperative relaying. 
In this paper, we study the spatial and temporal correlation of interference in mobile ad-hoc networks 
under a correlated shadowing environment. 
By modeling the node locations as a Poisson point process with an i.i.d. mobility model
and considering Gudmundson (1991)'s spatially correlated shadowing model, 
we theoretically analyze the relationship between the correlation distance of log-normal 
shadowing and the spatial and temporal correlation coefficients of interference. 
Since the exact expressions of the correlation coefficients are intractable, 
we obtain their simple asymptotic expressions as the variance of log-normal shadowing increases. 
We found in our numerical examples that 
the asymptotic expansions can be used as tight approximate formulas and 
useful for modeling general wireless systems under spatially correlated shadowing. 
\end{abstract}

}

\maketitle

\IEEEdisplaynontitleabstractindextext

%
\IEEEpeerreviewmaketitle

\IEEEraisesectionheading
{\section{Introduction}\label{sec:introduction}}
\IEEEPARstart{I}{n} wireless mobile systems, interference directly relates to the communication quality of mobile users and thus is a key factor in the performance evaluation, design, and management of the systems. In general, interference can be spatially and temporally correlated because of the correlations of the node locations and fading effect among a transmission area and multiple time-slots. Due to the spatio-temporal aspects of mobile wireless communications, such correlation of interference may affect their performance, e.g., end-to-end throughput, retransmissions~\cite{Haen13a}, cooperative relaying~\cite{Tanb13}, broadcast communications, multiple antennas~\cite{Tanb14}, and handovers~\cite{Kris17}. For example, if interference is temporally correlated, quick retransmission after an outage is likely to fail. Similarly, outages at closely located multiple-antennas may be correlated due to the spatially correlation of interference, which leads to frequent simultaneous outages. Thus, the correlation of interference needs to be analyzed to efficiently design various wireless systems. 

Typically, interference in wireless mobile systems is varied due to the randomness of fading  or shadowing effects, traffic, and spatial locations of transmitters/receivers including mobility. However, these characteristics are difficult to analyze comprehensively with a simulation or experimental approach because we must consider so many factors and parameter settings that such approaches do not scale. To solve this problem, a probabilistic modeling approach, stochastic geometry, has been taken for the correlation analysis of interference in the past decade~\cite{Gant09, Schi12, Gong14, Wen15, Wen16, Kris17, Kouf16, Kouf17, Kouf18}. In this approach, node locations are modeled by a point process, such as a Poisson point process (PPP), and fading and the medium access control (MAC) behavior are also statistically modeled. Ganti and Haenggi~\cite{Gant09} first studied the spatio-temporal correlation of interference in an ad-hoc network by assuming that node locations are distributed with a PPP. They theoretically showed that even with ALOHA as the MAC, the temporal correlation of node locations induces that of interference, which results in that of outage. Furthermore, Schilcher et al.~\cite{Schi12} showed that the temporal correlation of fading or traffic can increase that of interference. 

Although the correlation of interference has been well studied in the literature, few researchers considered the {\it spatially correlated shadowing}. Indeed, most previous work assumed i.i.d. fading. More specifically, only Rayleigh fading is often assumed because of its mathematical tractability. However, shadowing (i.e., slow fading) is spatially correlated in a scale of 50--100 m~\cite{Gold05} because of the effects of reflection or obstacles in transmission channels. In addition, a widely accepted model for the spatial correlation of shadowing was proposed by Gudmundson~\cite{Gudm91}. In this model, shadowing in a channel between a base station (BS) and a moving user is modeled by a log-Gaussian process with an exponentially decaying auto-correlation function that depends on the moving distance. According to this model, the distance at which the correlation appears (disappears) is commonly called a {\it correlation (decorrelation) distance}. Due to the spatial correlation of shadowing, the interference received at adjacent receivers may be more correlated compared with models with the i.i.d. fading assumption. However, such impacts have been ignored, and the relationship between the spatial or temporal correlation of interference and the correlation distance has not been studied. 

In this paper, we study the spatial and temporal correlation of interference in a mobile ad-hoc network (MANET) under a spatially correlated shadowing environment. We model the node locations by a PPP and the correlated shadowing effect by Gudmundson's widely accepted model~\cite{Gudm91}. On the basis of a stochastic geometry approach, we first derive the spatial correlation coefficient of interference, i.e., interferences received at different locations. Furthermore, by considering an i.i.d. mobility model of nodes, we next derive the temporal correlation coefficient of interference and show how the mobility affects the correlation of interference in a correlated shadowing environment. Since the exact expressions of both the spatial and temporal correlation coefficients have intractable  forms, we present their asymptotic expansions when the {\it variance} of the shadowing is large by using Watson's lemma (see e.g., \cite{Mill06}). 
%
As far as we know, this is the first study to use Watson's lemma to theoretically analyze spatially correlated shadowing in wireless networks.
%
The obtained asymptotic expansions are expressed as simple closed-form formulas, and so they show readable relationships between the correlation distance of the shadowing, the correlation of interference, and other system parameters. Our results show that the correlation distance of interference is smaller than that of shadowing. In addition, the temporal correlation of interference mainly depends on the probability that a node stays at the same position and does not converge to 0 even in a very high-mobility environment. Furthermore, we found that the node density increases the temporal correlation of interference. We also found in our numerical examples that the asymptotic expansions can be used as tight approximate formulas of the correlation coefficients of interference and so are useful for the modeling and performance evaluation of various wireless systems in a spatially correlated shadowing environment. 

The reminder of this paper is organized as follows. In Section~\ref{sec-related}, we summarize related work. Section~\ref{sec-model} explains our model considered in this paper. Section~\ref{sec-ana} presents the results for the spatial and temporal correlations of interference. Section~\ref{sec-num} provides several numerical examples. 
Finally, we conclude this paper in Section~\ref{sec-conclude}. 

A part of this paper was presented in \cite{Kimu18}, in which only the temporal correlation
coefficient of interference was analyzed in a 1-D network. 

\section{Related Work}
\label{sec-related}
\noindent 
{\bf Analysis of correlation coefficients of interference} Since the correlation of interference greatly affects the spatial and temporal characteristics of various wireless systems, its stochastic-geometric analysis has been studied under various settings. Ganti and Haenggi~\cite{Gant09} first studied the spatio-temporal correlation coefficients of interference in ad-hoc networks. Similar to other researchers, they modeled the node locations by a PPP under an i.i.d. fading assumption. Schilcher et al.~\cite{Schi12} extended Ganti and Haenggi's work~\cite{Gant09} by considering Rayleigh block fading and the duration of transmissions. Mobility of nodes were studied by Gong and Haenggi~\cite{Gong14}. They derived the temporal correlation coefficient of interference assuming i.i.d. mobility models for transmitters, such as Brownian motion and random walk models. Two extensions for node location models were considered by Wen et al.: $K$-tier heterogeneous networks~\cite{Wen15} and a clustered or a repulsive point process~\cite{Wen16}. 
Recently, Schilcher et al.~\cite{Schi19} studied the autocorrelation of interference in Poisson networks
as an extension of their previous work~\cite{Schi12}. 
Assuming Nakagami block fading and Rayleigh fading in accordance with Clarke's model, 
the authors derived the coherence time of interference, that is, the time lag
until the interference correlation becomes dips below a certain threshold. 
Whereas the above studies considered ad-hoc networks, Krishnan and Dhillon~\cite{Kris17} studied the spatial and temporal correlation of interference and derived the joint coverage probability in a cellular network with a closest BS association policy, in which users have chances to hand off to a new serving BS while they are moving. In addition, Koufos and Dettmann~\cite{Kouf18} studied the case where the node mobility is correlated by considering a random waypoint mobility model in a 1-D finite lattice. 

The above studies assumed i.i.d. fading for each transmission channel and ignored the impact of correlated shadowing. Most recently, Koufos et al.~\cite{Kouf16, Kouf17} analyzed the effect of blockage on the temporal correlation of interference. In their latest work \cite{Kouf17}, they modeled the locations of nodes and obstacles by 1-D PPPs and consider the penetration loss of obstacles. However, they did not analyze the relationship between the correlation distance of shadowing and correlation of interference in a general framework. In addition, the obtained results are complicated, and so the impacts of various system parameters were not clearly shown even though their modeling assumption was quite simple. 

The effect of blockage has also recently been studied by Aditya et al.~\cite{Adit17, Adit18} for a localization network, 
in which randomly located anchors (transceivers) attempt to detect a target in the presence of obstacles. 
The authors modeled the obstacles with a Poisson line process in \cite{Adit17} and a germ-grain model in \cite{Adit18}
to derive the blockage probability that there is not enough anchors with
line-of-sight (LoS) paths to the target (i.e., paths not interrupted by obstacles). 

\noindent 
{\bf Spatio-temporal analysis of wireless systems} Not only have correlation coefficients of interference been analyzed but also other spatio-temporal aspects of wireless mobile systems have been studied. For example, local delay, defined as the mean number of time slots required for successful transmission of a packet, was first studied by Baccelli and B{\l}aszcyzyszyn~\cite{Bacc10} for MANETs. On the other hand, Gulati et al.~\cite{Gual12} considered a model where user traffic has a certain duration. Gong and Haenggi~\cite{Gong13} derived the local delay in MANETs with i.i.d. mobility models. In addition, that in D2D networks was recently studied by Salehi et al.~\cite{Sale17}. Furthermore, Haenggi~\cite{Haen12} derived the joint successful transmission probability at multiple antennas. Tanbourigi et al.~\cite{Tanb14} also modeled multiple antennas with maximum ratio combining and derived the joint coverage probability. 
Crismani et al.~\cite{Cris15} considered a decode-and-forward relaying system 
in which consecutive transmission attempts are temporally correlated. 
Afify et al.~\cite{Afif16} presented a unified mathematical framework for cellular networks with 
multiple-input-multiple-output (MIMO) and studied the temporal correlation in retransmissions. 
The effect of the spatial correlation of interference on opportunistic secure information transfer was 
analyzed in~\cite{Ali17}.
However, all these studies assumed i.i.d. Rayleigh fading to maintain mathematical tractability.

\noindent 
{\bf Modeling log-normal shadowing} Since the shadowing is spatially correlated~\cite{Gold05}, several analytical models for it have been proposed~\cite{Gudm91, Graz02, Agra09}. Gudmundson~\cite{Gudm91} proposed a widely accepted model in which the shadowing between a BS and moving user is modeled by an autoregressive process with an exponentially decaying autocorrelation function. Several extensions of this model have also been proposed: the case of multiple BSs~\cite{Graz02} and multi-hop networks~\cite{Agra09}. The correlation distance or other parameters of this model have also been experimentally studied~\cite{Gudm91, Weit02, Baek16}. Baek et al.~\cite{Baek16} recently showed that correlation distances for the mm-wave frequency bands are similar to those for other lower frequency bands. 

Since a log-normal shadowing model makes theoretical interference analysis intractable, a few studies aimed to efficiently model or approximate the shadowing effect. The mean interference and path-loss process in cellular networks when the variance of the shadowing increases were studied by B{\l}aszcyzyszyn and Karray \cite{Blas13a} and B{\l}aszcyzyszyn et al. \cite{Blas13b}, respectively. Heath et al.~\cite{Heat13} approximated the log-normal shadowing by gamma distributions with the same first and second moments. Whereas the above studies assumed i.i.d. shadowing for each transmission channel, Renzo~et al.~\cite{Renz13} considered a model where all channels are equicorrelated. Baccelli and Zhang~\cite{Bacc15} recently proposed a correlated shadowing model in which correlated log-normal shadowing is approximated by a random variable depending on the number of buildings that a transmission channel penetrates. Recently, this idea was applied to a 3-D network model by Lee et al.~\cite{Lee16}. However, the impact of node mobility was not considered.

\section{Model Description}
\label{sec-model}
\subsection{Spatial and mobility model}
In this section, we explain our spatial model. We consider $n$-dimensional space $\bbR^n$ $(n = 1, 2)$ and assume that a target receiver is located at the origin $o$. Nodes, i.e., potential transmitters, are randomly distributed on $\bbR^n$. Note that we consider mobility of nodes, but the receiver is assumed to be fixed at the origin. In this paper, we consider two types of network models: a line model ($n=1$) and planar model ($n = 2$). In a line model, i.e., $n=1$, the model can be applied to vehicular networks on highways whereas general cellular or wi-fi networks can be modeled in a planar model, i.e., $n = 2$. We assume that time-space is slotted and nodes independently move in each time-slot. More precisely, the locations of nodes at the beginning of time slot $t \in \bbZ_+ \triangleq \{0,1,2,\dots\}$ are modeled by a PPP $\Phi(t)$ with intensity $\lambda$. Let $x_i(t) \in \Phi(t)$ $(i \in \bbZ_+)$ on $\bbR^n$ denote the position of the $i$-th node (a vehicle, mobile user, BS) at the time slot $t$. Furthermore, we assume an i.i.d. mobility model as follows: the moving vector of the node $i$ during the time slot $t$, which is given as $v_i(t) \triangleq x_i(t+1) - x_i(t)$ $(t \in \bbZ_+)$, is assumed to be independently distributed with a probability density function $\psi(\cdot)$, 
which does not depend on $x_i(t)$ and satisfies $\psi(0) > 0$. 
Due to the displacement theorem of PPPs (see e.g., Theorem 2.33 in \cite{Haen13}), the realization of $\Phi(t)$ at fixed time slot $t$
remains a homogeneous PPP if $\Phi(0)$ is homogeneous. 
For simplicity, we assume that $\psi(\cdot)$ is rotation invariant, i.e., $\psi(v) \equiv \psi(\|v\|)$ for $v \in \bbR^n$. Here, $\|\cdot\|$ denotes the Euclidean norm. 

\subsection{Channel and MAC modeling}
We next explain our channel model. First, we assume that all transmitters have the unit transmission power. The path loss model is assumed as, for distance $r \ge 0$,
\[
\ell(r) = \varepsilon_0 + r^{\alpha}. 
\]
Here, $\alpha > n$ ($n = 1,2$) is a path loss exponent
 and $\varepsilon_0 > 0$ is a parameter for avoiding singularity at $r = 0$. Thus, the received power (i.e., interference) at $o$ from the $i$-th node at the time slot $t$ can be represented as 
\[
{h_i \over \ell(\|x_i(t)\|) },
\]
where $h_i$ denotes a copy of shadowing variable $h$ corresponding to the transmission channel of the $i$-th node. In this paper, we only consider the effect of the shadowing  and do not take into account the multi-path fading (fast fading), such as Rayleigh fading. However, such a multi-path fading effect is easy to include because it can be considered as i.i.d. for each channel. In accordance with a widely accepted assumption for the shadowing, the shadowing variable is characterized with a log-Gaussian random process with parameter $\sigma_{\dB}$. In other words, for any given node, the corresponding shadowing variable $h$ is log-normally distributed as
\[
h = \exp\left( - \sigma^2_{\dB}/2 + \sigma_{\dB} Z\right),
\]
where $Z \sim \calN(0, 1)$ and $\sigma_{\dB}$ is considered as the variance of shadowing. Note that 
\begin{equation}
\bbE_0[h] = 1, \qquad \bbE_0[h^2] = e^{\sigma_{\dB}^2}.
\label{eq-h-moments}
\end{equation}
Here, the above expectations are taken with respect to a typical node, i.e., under the condition in which a node exists. In addition, if we consider two nodes $i$ and $j$ $(i \neq j)$ and their transmissions to $o$, the corresponding shadowing variables $h_i$ and $h_j$ have the following correlation coefficient in the logarithmic sense:
\begin{equation}
\rho_{\dB} \triangleq
{\bbE_{ij}[\ln h_i \ln h_j ] \over \sigma^2_{\dB}}
 = e^{- {\|x_i - x_j\| \over d_{\mathrm{cor}}}\ln 2},
\label{eq-lognorm-corr-0}
\end{equation}
where $d_{\mathrm{cor}}$, called the correlation (decorrelation) distance, 
represents the distance at which the correlation coefficient $\rho_{\dB}$ is equal to 0.5~\cite{Gudm91}
\footnote{$d_{\mathrm{cor}}$ depends on the environment. A typical value of $d_{\mathrm{cor}}$ in urban environments is [50, 100] m~(see e.g., \cite{Gudm91}).}. In addition, the expectation $\bbE_{ij}$ represents the expectation with respect to two typical nodes. For simplicity, we write $d_0 \equiv d_{\mathrm{cor}} / \ln 2$ hereafter.
%
Since each $h_i$ is log-normally distributed
with mean and variance given by (\ref{eq-h-moments}), 
the expression (\ref{eq-lognorm-corr-0}) is equivalent to 
%
%
\begin{equation}
\bbE_{ij}[h_i h_j] = 
\exp\left( \sigma_{\dB}^2 \rho_{\dB} \right)
=
\exp\left( \sigma_{\dB}^2 e^{- {\|x_i - x_j\| \over d_{0}}}\right).
\label{eq-lognorm-corr}
\end{equation}
This is the widely accepted model proposed by Gudmundson~\cite{Gudm91} for the spatial correlation of shadowing. 
In addition, in this paper, we assume that the shadowing effect depends only on the location and is time-invariant.  
In other words, the value of the shadowing variable at a fixed node location does not change over time.

We assume that all transmitters use slotted-ALOHA as the MAC protocol, i.e., each node determines whether to transmit radio waves or not with probability $p$ in each transmission time-slot. $\calS_i(t)$ denotes the indicator function that equals $1$ when the transmitter $i$ is transmitting radio waves at the beginning of time slot $t$, 0 otherwise. Then, $S_i(t)$ can be considered as an independent Bernoulli random variable with a parameter $p$. By definition, the total interference power received at the origin equals 
\begin{equation}
I(t) =\sum_{x_i(t) \in \Phi(t)} {h_i(t) \calS_i(t) \over \ell(\|x_i(t)\|) }. 
\label{eq-def-I}
\end{equation}

Note that the transmission time-slot is different from the mobility time-slot. For simplicity, we assume that the transmission time-slot is fixed and smaller than the mobility time-slot. In other words, the transmission of each vehicle at time-slot $t$ does not continue until time slot $t+1$, and thus $\calS_i(t)$ and $\calS_i(t+1)$ are independent for all $i$ and $t$.

Although we can consider other MAC protocols for realistic scenarios, 
such as CSMA, we assume slotted-ALOHA as the MAC to maintain tractability. 
However, such a simple model can be used as a baseline for more general cases. 
It is reported that if node density increases, 
the behavior of CSMA tends to mimic that of slotted ALOHA (e.g., see \cite{Nguy13}). 
Furthermore, Tong et al.~\cite{Tong16} demonstrated that the performance of an ALOHA-type model
is similar to that obtained by an NS2 simulation with a CSMA model. 

\subsection{Correlation coefficients of interference}
Finally, we define the spatial and temporal correlation coefficients of interference. To consider the spatial correlation of interference, we assume another receiver is located at $\tilde{o}_{\delta}$ on the $x$-axis such that $\|\tilde{o}_{\delta} - o\| = \delta > 0$ and consider the correlation between the interferences received at $o$ and $\tilde{o}_\delta$. Let $I_0 \equiv I_0(t)$ and $\tilde{I}_\delta \equiv \tilde{I}_\delta(t)$ denote the interferences received at $o$ and $\tilde{o}_\delta$, respectively. Here, we omit the parameter $t$ because we consider a certain fixed time slot (e.g., $t = 0$) in the analysis of spatial correlation of interference. The {\it spatial covariance} of interference is then formally defined as the covariance of $I_0$ and $\tilde{I}_\delta$, i.e.,
\[
\Cov[I_0,\tilde{I}_\delta] = \bbE[I_0 \tilde{I}_\delta] - \bbE[I_0]\bbE[\tilde{I}_\delta]
= \bbE[I_0 \tilde{I}_\delta] - (\bbE[I])^2,
\]
Furthermore, the {\it spatial correlation coefficient} of interference is defined as
\begin{equation}
\rho_{\rsp,\delta} \triangleq {\Cov[I_0,\tilde{I}_\delta] \over \sqrt{\Var [I_0] \Var[\tilde{I}_\delta]}}
= 
{ \bbE[I_0 \tilde{I}_\delta] - (\bbE[I])^2
\over
\Var[I]
}.
\label{eq-def-rho-sp}
\end{equation}

We next define the temporal correlation as the correlation between the interferences received at the origin among different two time-slots. More precisely, we define the {\it temporal covariance of interference } for time interval $\tau$ slots ($\tau \in \bbZ_+$) as

\begin{eqnarray}
\Cov[I(t), I(t+\tau)] &=& \bbE[I(t)I(t+\tau)] - \bbE[I(t)]\bbE[I(t+\tau)]
\nonumber
\\
&=&
\bbE[I(t)I(t+\tau)] - (\bbE[I])^2.
\label{eq-def-cov-temp}
\end{eqnarray}
Similarly, the {\it temporal correlation coefficient of interference} is defined as 

\begin{eqnarray}
\rho_{\rtm,\tau} &\triangleq& {\Cov[I(t),I(t+\tau)] \over \sqrt{\Var [I(t)]\Var[I(t+\tau)]}}
\nonumber
\\
&= &
{ \bbE[I(t)I(t+\tau) ] - (\bbE[I])^2
\over
\Var[I]
}.
\label{eq-def-rho_tm}
\end{eqnarray}

\section{Spatial and Temporal Correlation Analysis of Interference}\label{sec-ana}
In this section, we present our main results,  the spatial and temporal correlation coefficient of interference. To do this, we first derive the first and second moments of interference, which will be used for the derivation of other results. Unlike the first moment of interference, the second moment does not have an explicit form because of the exponentially decaying cross correlation of shadowing variables (see (\ref{eq-lognorm-corr})). To obtain the readable relationship between the shadowing parameters and the second moment of interference, we use Watson's lemma (see e.g., \cite{Mill06}) and derive simple but useful asymptotic expansions of the second moment of interference when $\sigma_{\dB} \to \infty$. The result can be directly applied to the derivation of the correlation coefficients. In addition, the obtained expansions can be used as tight approximate formulas with a realistic value of $\sigma_{\dB}$, which will be shown in our numerical examples. 

\subsection{Preliminary}\label{subsec-pre}
For later use, we first introduce a useful mathematical tool, known as Watson's lemma, which plays an 
important role in our analysis. In short, Watson's lemma gives an asymptotic expression 
of the following form of an exponential integral when $\sigma$ is large:
\begin{equation}
\widetilde{F}(\sigma) 
= \int_a^{a + \delta} e^{\sigma R(t)} g(t) \rmd t
=
\int_0^{\delta} e^{\sigma R(a + u)} g(a + u) \rmd u,
\label{eq-tilde-F}
\end{equation}
where $\delta > 0$ and $R(t)$ has its maximum at $t= a$. 
Moreover, let $\widetilde{R}(t):= R(t) - R(a)$, which has zero at $t = a$. 
As shown later, the above integral can be used to calculate terms related to the cross-correlation of shadowing variables
(see e.g., (\ref{eq-lognorm-corr})). 
%
Although Watson's lemma itself is a classical and fundamental result, 
this paper uniquely applies it to a theoretical analysis of 
spatially correlated shadowing in wireless networks. 

%
\begin{prop}[Miller~\cite{Mill06} (Section 3.3)]\label{prop-watson}
Suppose that $R(t)$ and $g(t)$ have an infinite number of continuous derivatives for $a \le t < a + \delta$ and 
$\tilde{R}'(a) = R'(a) < 0$. We then have the following asymptotic expansion:
\begin{eqnarray}
\widetilde{F}(\sigma) = \sum_{n=0}^\infty {\varphi^{(n)}(0) \over \sigma^{n+1}},
\qquad \mbox{as $\sigma \to \infty$ with $\sigma > 0$},
\label{eq-watson-1}
\end{eqnarray}
%
where $\varphi^{(n)}(s)$ denotes the $n$-th derivative of $\varphi(s)$ and 
\begin{equation}
\varphi(s) = g(a + u(s))u'(s),
\label{eq-def-psi(s)}
\end{equation}
where $u(s)$ is a function of $s$ that satisfies $u(0) = 0$ and 
\begin{equation}
\widetilde{R}(a + u(s)) = -s.
\label{eq-def-u}
\end{equation}
%
\end{prop}
%

Using the above expressions, we can obtain the following corollary, which produces
an asymptotic formula for $\widetilde{F}(\sigma)$ in a special case. 
%
\begin{coro}\label{lem-watson}
%
If $g(t)$ satisfies the following condition for $n \in \bbN$,  
\begin{align}
&g'(a) = g''(a) = \cdots = g^{(n-2)}(a) = 0, 
\label{eq-g-cond-1}
\\
&g^{(n-1)}(a)\neq 0,\quad g^{(n)}(a) \neq 0,
\label{eq-g-cond-2}
\end{align}
then, 
\begin{align}
&\widetilde{F}(\sigma) 
=
e^{\sigma R(a)}
\left[
{g^{(n-1)}(a) \over (- R'(a))^{n}} {1 \over \sigma^{n}}
+
\left\{
{g^{(n)}(a) \over (- R'(a))^{n+1}}
\right.
\right.
\nonumber
\\
&+
\left.\left. 
{n ( n+1 ) \over 2}
{R''(a) g^{(n-1)}(a) \over (- R'(a))^{n+2}} 
\right\}
{1 \over \sigma^{n+1}}
+
O\left(
{1 \over \sigma^{n+2}}
\right)
\right].
\label{eq-watson-n>2}
\end{align}
%
\end{coro}
\begin{proof}
The proof of this corollary is given in Appendix~\ref{appen-proof-of-lem-watson}. 
\end{proof}

\subsection{First and second moments of interference}
We next derive the first and second moments of interference. 
In what follows, when we only consider a certain time slot, 
we omit the parameter $t$ representing a time slot, such as $x_i \equiv x_i(t)$ and $I \equiv I(t)$. 
Due to the definition of shadowing variable $h$, the marked point process $\hat{\Phi} = \{(x_i, h_i)\}$ is stationary. In addition, the mean mark of a typical point is given by (\ref{eq-h-moments}). Therefore, by applying Campbell's theorem~(see e.g., \cite{Schn08}) to (\ref{eq-def-I}), the mean interference can be obtained as follows:
\begin{eqnarray}
\bbE[I] &=&
\bbE\left[ \sum_{x_i \in \Phi} {h_i \calS_i \over \ell(\|x_i\|)}\right]
=
\int_{\bbR^n}
{\lambda p \bbE_0[h] \over \eps_0 + \|x\|^{\alpha}}\rmd x
\nonumber
\\
&=&
{\lambda p n V_n \pi 
\over
\eps_0^{1 - {n \over \alpha}} \alpha \sin(n \pi / \alpha)}
,
\label{eq-mean-I}
\end{eqnarray}
%
where $V_n$ represents the volume of an $n$-dimensional unit ball. 
We next consider the second moment of the interference, $\bbE[I^2]$. By definition, we have
\begin{eqnarray}
\bbE[{I^2}] 
&=&
\bbE \left[\sum_{x_i, x_j \in \Phi_{\neq}^{(2)}}\!\!
{h_i h_j \calS_i\calS_j\over \ell(\|x_i\|) \ell(\|x_j\|)}
+
\sum_{x_i \in \Phi} {h_i^2 \calS_i\over (\ell(\|x_i\|))^2 }\right],
\nonumber
\label{eq-E[I^2]-0}
\end{eqnarray}
where $\Phi^{(2)}_{\neq}$ denotes all the sets of the distinct pairs in $\Phi$. Recall that $\bbE_0[h^2] = e^{\sigma_{\dB}^2}$ (see (\ref{eq-h-moments})) and the cross correlation of shadowing variables is given by (\ref{eq-lognorm-corr}). Therefore, by applying Campbell's theorem (e.g., \cite{Schn08}) to the above, we obtain
\begin{eqnarray}
\bbE[I^2] 
&=&
 \int\!\!\!\int_{(\bbR^n)^2} \!
{\lambda^2 p^2\bbE_{ij}[h_i h_j] \over \ell(\|x_i\|) \ell(\|x_j\|)} \rmd x_i \rmd x_j
+
\int_{\bbR^n} 
{\lambda p\bbE_0[h^2] \over (\ell(\|x\|))^2} \rmd x
\nonumber
\\
&=&
 \int\!\!\!\int_{(\bbR^n)^2} \!\!\!
{\lambda^2p^2 
e^{\sigma_{\dB}^2 e^{- \|x_i - x_j\| \over d_0}}
\over \ell(\|x_i\|) \ell(\|x_j\|)} \rmd x_i \rmd x_j
+
\lambda p e^{\sigma^2_{\dB}}\gamma_n,
\quad~
\label{eq-E[I^2]-1}
\end{eqnarray}
where 
%
\begin{equation}
\gamma_n
\triangleq
\int_{\bbR^n} 
{ \rmd x \over (\eps_0 + \|x\|^\alpha)^{2} }
=
{n V_n \pi  (\alpha - n) 
\over
\eps_0^{2 - {n \over \alpha}} \alpha^2 \sin(n \pi / \alpha)}.
\label{eq-E[I^2]-sec-term}
\end{equation}
%
The first term in (\ref{eq-E[I^2]-1}) does not have a closed-form even though it can be numerically computed. 
However, we can show a simple asymptotic expansion of this integral as $\sigma_{\dB} \to \infty$ by using Watson's lemma 
(see Section~\ref{subsec-pre}). 
To proceed, we provide the following additional result.  
\begin{prop}\label{lem-appen-asym}
Let $f(\cdot)$ denote an arbitrary function on $\bbR$ such that $f(0) \neq 0 $ and $f'(0)$ exists. In addition, let 
\begin{equation}
\Psi_n (f) \equiv 
\int\!\!\!\int_{(\bbR^n)^2} 
{\exp\left(\sigma^2 e^{- \|y\| \over d_0}\right) f(\|y\|)
\over \ell(\|x\|) \ell(\|x + y\|)} \rmd y \rmd x.
\label{eq-def-psi}
\end{equation}
We then have, as $\sigma \rightarrow \infty$: i) if $n = 1$, 
\begin{align}
\Psi_1(f)
&\sim
2 d_0 e^{\sigma^2} 
\left[
\gamma_1
\left(
 {f(0) \over \sigma^2} 
+
{f(0) + d_0 f'(0) \over \sigma^4}
\right)
\right.
\nonumber
\\
&
\left.
~~~+ 
{ d_0 f(0) \over \eps_0^2 \sigma^4}
+
O\left({1 \over \sigma^6}\right)
\right],
\label{eq-int-I^2-sub-last-n=1}
\end{align}
and ii) if $n \ge 2$, 
%
\begin{align}
&\Psi_n(f)
\sim
n d_0^n  V_n e^{\sigma^{2}} \gamma_n 
\left\{
 {(n - 1)! f(0) \over \sigma^{2n}} 
 \right.
 \nonumber
 \\
 &
 \left.
+
\left(
{(n + 1)! f(0) \over 2}
+ 
d_0 n! f'(0)
\right)
 {1 \over \sigma^{2(n+1)}} 
+
O\left(
{1 \over \sigma^{2(n+2)}}
\right)
\right\},
\label{eq-int-I^2-sub-last-n=2}
\end{align}
%
where $\gamma_n$ is given in (\ref{eq-E[I^2]-sec-term}). 
\end{prop}
%
\begin{proof}
The proof of this proposition is given in Appendix~\ref{appen-proof-of-lem-appen-asym}. 
\end{proof}
%

\begin{rem}
As shown in Proposition~\ref{lem-appen-asym},
the asymptotic expansion of $\Psi_n(f)$ results in 
slightly different forms for cases $n = 1$ and $n\ge2$
(see the third term in (\ref{eq-int-I^2-sub-last-n=1})). 
Since our asymptotic analysis of the spatial and
temporal correlations of interference is based on 
Proposition~\ref{lem-appen-asym}, 
our main results presented later 
also have different forms in cases $n = 1$ and $n\ge2$. 
This discrepancy can be attributed to the fact that 
Watson's lemma requires the differentiability of the related function $g(t)$
at the maximum point of the exponential function
(see Proposition~\ref{prop-watson} and Corollary~\ref{lem-watson}). 
Since the expression of this function is different in case $n=1$
to avoid non-differentiability, the resulting asymptotic expansion also has 
a different expression. 
\end{rem}

Using the above result, we have the following simple asymptotic expansions for $\bbE[I^2]$ 
in a strong shadowing environment, i.e., $\sigma_{\dB}$ is large. 
\begin{lem}\label{thm-second-moment}
The second moment of interference has the following asymptotic expansions as $\sigma_{\dB}^2 \to \infty$: i) if $n = 1$, 
\begin{eqnarray}
\bbE[I^2] 
&\sim&
\lambda p e^{\sigma_{\dB}^2} 
\left[
\gamma_1
\left(
1
+
{2 \lambda p d_0 \over \sigma_{\dB}^2} 
+
{2 \lambda p d_0 \over \sigma_{\dB}^4}
\right)
\right.
\nonumber
\\
&{}&
\left.
+
{2 \lambda p d_0^2 \over \eps_0^2 \sigma_{\dB}^4}
+ O\left({1 \over \sigma_{\dB}^6}\right)
\right]
\label{eq-app-EI^2-n=1},
\end{eqnarray}
ii) if $n = 2$, 
\begin{equation}
\bbE[I^2]
\sim
\lambda p e^{\sigma_{\dB}^2} \gamma_2
\left[
1 + 
{2\pi \lambda p d_0^2 \over \sigma_{\dB}^4}
+
{6 \pi \lambda p d_0^2 \over \sigma_{\dB}^6}
+ O\left({1 \over \sigma_{\dB}^8}\right)
\right].
\label{eq-app-EI^2-n=2}
\end{equation}
\end{lem}
\begin{proof}
The first integral term in the right-hand side of the second equation in (\ref{eq-E[I^2]-1}) can be rewritten as,
for $n = 1,2$,
\begin{eqnarray}
\lefteqn{
\int\!\!\!\int_{(\bbR^n)^2} 
{\exp\left(\sigma_{\dB}^2 e^{- \|x_i - x_j\| \over d_0}\right)
\over \ell(\|x_i\|) \ell(\|x_j\|)} \rmd x_i \rmd x_j
}&& \qquad
\nonumber
\\
&=&
\int_{\bbR^n} {1 \over \ell(\|x\|) }
\int_{\bbR^n} 
{\exp\left(\sigma_\dB^2 e^{- \|y\| \over d_0}\right) \over \ell(\|x + y\|)} \rmd y \rmd x.
\label{eq-lem-appen-1st}
\end{eqnarray}
Thus, substituting $f(x) \equiv 1$ into (\ref{eq-int-I^2-sub-last-n=1}) and (\ref{eq-int-I^2-sub-last-n=2}) and combining them with (\ref{eq-E[I^2]-1}), (\ref{eq-E[I^2]-sec-term}), and (\ref{eq-lem-appen-1st}), we can readily obtain (\ref{eq-app-EI^2-n=1}) and (\ref{eq-app-EI^2-n=2}). 
\end{proof}

\begin{rem}\label{rem-var}
According to Lemma~\ref{thm-second-moment} and (\ref{eq-mean-I}), we can also obtain asymptotic expansions of the variance of the interference $\Var[I]$. Recall here that (see (\ref{eq-mean-I})), for $n = 1, 2$, 
\begin{equation}
(\bbE[I])^2 / e^{\sigma^2_{\dB}} = O( e^{- \sigma^2_{\dB}}).
\label{eq-mean-order}
\end{equation}
Thus, the asymptotic expansions of the variance as $\sigma_{\dB} \to \infty$ have the same forms as those in (\ref{eq-app-EI^2-n=1}) and (\ref{eq-app-EI^2-n=2}), i.e., $\bbE[I^2] \sim \Var[I]$. 
\end{rem}

\begin{rem}
Since Proposition~\ref{lem-appen-asym} gives an
asymptotic formula of $\Psi_n(f)$ for $n \in \bbN$, 
we can consider case $n = 3$. 
Such a scenario can be applied to unmanned-aerial-vehicle (UAV)
networks, in which UAVs act as BSs or user equipment~\cite{Moza18}. 
However, in this case, the positions of UAVs are physically limited 
by the ground or maximum altitude of UAVs. 
Thus, we must consider a certain closed region (not $\bbR^3$)
that depends on the position of the receiver in 3D-space. 
In addition, it has been reported that 
the channel conditions, such as LoS or NLoS, 
and related path-loss models differ from ground communications
depending on the altitude of the transmitters~\cite{Moza18}. 
Hence, a more detailed configuration is necessary for realistic modeling
in case $n = 3$. 
In this paper, we only consider 
cases of $n=1$ and $n=2$ to focus on a general setting. 
\end{rem}

%

%
\subsection{Spatial correlation of interference}\label{subsec-spa-cor}
On the basis of the results in the previous section, we next derive the spatial correlation coefficient of interference. Similar to the second moment of interference, the mean product of $I_0$ and $\tilde{I}_\delta$ can be rewritten as 
\begin{eqnarray}
\bbE[I_0 \tilde{I}_\delta]
&=&
\bbE \left[
\sum_{x_i \in \Phi} {h_i \calS_i \over \ell(\|x_i\|)}
\sum_{x_i \in \Phi} {\tilde{h}_i \calS_j \over \ell(\|x_i - \tilde{o}_\delta\|)}
\right]
\nonumber
\\
&=&
\bbE \left[\sum_{x_i, x_j \in \Phi^{(2)}_{\neq}} {h_i \tilde{h}_j \calS_i\calS_j 
\over
\ell(\|x_i\|) \ell(\|x_j- \tilde{o}_\delta\|)}
\right]
\nonumber
\\
&{}&
+
\bbE \left[
\sum_{x_i\in \Phi} {h_i \tilde{h}_i \calS_i 
\over
\ell(\|x_i\|)\ell(\|x_i - \tilde{o}_\delta\|) }
\right],\quad~
\label{eq-spa-corr-0}
\end{eqnarray}
where $\tilde{h}_i$'s denote the shadowing variables corresponding to the channel between the $i$-th node and receiver $\tilde{o}_\delta$. Recall here that we only assumed in (\ref{eq-lognorm-corr}) the correlation coefficient of shadowing for two different transmitters and a receiver at $o$. Thus, to consider the cross-correlation of $h_i$ and $\tilde{h}_j$ ($i \neq j$) i.e., the case of a pair of transmitters and a pair of receivers, we set the following simple assumption proposed by Wang et al.~\cite{Wang08} throughout this subsection. 
\begin{assumpt}\label{assu-corr}
Suppose a pair of transmitters are at $x_i$ and $x_j$ and a pair of receivers are at $o$ and $\tilde{o}$, which are different points on $\bbR^n$ ($n=1,2$). If $h_i$ (resp. $\tilde{h}_j$) is the shadowing variable corresponding to the channel between $x_i$ and $o$ (resp. $x_j$ and $\tilde{o}$), then the correlation coefficient of $h_i$ and $\tilde{h}_j$ approximately equals 
\begin{equation}
{\bbE[\ln h_i \ln \tilde{h}_j] \over \sigma^2_{\dB}}= 
e^{- {\|o - \tilde{o}\| + \|x_i - x_j\| \over d_0}}.
\label{eq-assu-1}
\end{equation}
\end{assumpt}
%

Recall that, from Gudmundson's model in (\ref{eq-lognorm-corr}),
the cross correlations $\bbE[\ln h_i \ln h_j]$ and $\bbE[\ln h_j \ln \tilde{h}_j]$ 
are determined by the distances $\|x_i - x_j\|$ and $\|o - \tilde{o}\|$, i.e., 
\[
{\bbE[\ln h_i \ln h_j] \over \sigma^2_{\dB}} = e^{- {\|x_i - x_j\| \over d_{0}}},
\quad
{\bbE[\ln h_j \ln \tilde{h}_j]  \over \sigma^2_{\dB}}= e^{- {\|o - \tilde{o}\| \over d_{0}}}.
\]
Thus, Assumption~\ref{assu-corr} indicates that the correlation coefficient 
of $h_i$ and $\tilde{h}_j$ is expressed as the product of 
the cross correlations $\bbE[h_i h_j]$ and $\bbE[h_j \tilde{h}_j]$. 
In other words, if we consider that  a receiver moved from $o$ to $\tilde{o}$ and 
a transmitter moved from $x_i$ to $x_j$, 
these movements have independent and equal impacts on the cross correlation of 
$h_i$ and $\tilde{h}_j$. 
Note that Assumption~\ref{assu-corr} implicitly assumes 
that the equally impacted correlation of shadowing 
i.e., $\bbE[ h_j \tilde{h}_j]$,
is also determined by (\ref{eq-lognorm-corr}), similar to $\bbE[h_i h_j]$. 
This underlying assumption is considered as valid in MANETs because
all nodes are assumed to have the same configuration (e.g., antenna type and height) 
and communicate using a common channel~\cite{Wang08}, 
which differs from a cellular network scenario. 
Although Assumption~\ref{assu-corr} seems simple and is suited to our approach, 
Wang et al.~\cite{Wang08} demonstrated that it has sufficient accuracy 
via comprehensive ray-model simulations based on 
a realistic 3-D deterministic propagation model~\cite{Tame97, Tame01}. 

 On the basis of  Assumption~\ref{assu-corr}, we obtain the following asymptotic expressions for the spatial covariance of interference at the $o$ and $\tilde{o}_\delta$.
\begin{lem}\label{thm-corr-spatial}
Suppose $\delta < d_0 \log \sigma^2_{\dB}$ (i.e., $e^{\delta \over d_0} < \sigma^2_{\dB}$)\footnote{Typically, $\sigma_{\dB}$ takes the values in range 3--15 dB~\cite{Gold05}. In addition, we found that $\rho_{\rsp, \tau}$ takes a very small value when $\delta \ge d_0$. Thus, this condition is fulfilled in realistic settings.}. Furthermore, let 
\begin{eqnarray}
U_n(\delta) &=& \int_{\bbR^{n}}
{ \rmd x
\over \ell(\|x\|) \ell(\|x - \tilde{o}_\delta\|)},
\qquad n = 1,2.
\label{eq-def-U}
\end{eqnarray}
Then, the spatial covariance of interference for distance $\delta > 0$ has the following asymptotic expansions as $\sigma_{\dB}^2 \to \infty$: i) if $n=1$, 
\begin{eqnarray}
\lefteqn{
\Cov[I_0,\tilde{I}_\delta] \sim
\lambda p
e^{\sigma_{\dB}^2 e^{ {-\delta \over d_0}}}
} &&
\nonumber
\\
&\times&
\left[
U_1(\delta)
\left\{
1 + {2 \lambda p d_0 e^{\delta \over d_0} \over \sigma_{\dB}^2}
+
{2 \lambda p d_0e^{2\delta \over d_0} \over \sigma_{\dB}^4}
+
O\left({e^{3\delta \over d_0} \over \sigma_{\dB}^6}\right)
\right\}
\right.
\nonumber
\\
&+&
\left.
{ 2\lambda p d_0^2 e^{2 \delta \over d_0} \over \eps_0(\eps_0 + \delta^\alpha) \sigma_{\dB}^4}
\right]
+
e^{\sigma_{\dB}^2 e^{- {2\delta \over d_0}}}
{\lambda^2 p^2 d_0^2 e^{4 \delta \over d_0} \over \eps_0^{2} \sigma_{\dB}^4},
\label{eq-app-spatial-n=1}
\end{eqnarray}
ii) if $n = 2$, 
\begin{eqnarray}
\lefteqn{
\Cov[I_0, \tilde{I}_\delta] \sim
\lambda p
e^{\sigma_{\dB}^2 e^{ {-\delta \over d_0}}}
U_2(\delta)
} &&
\nonumber
\\
&\times&
\left[
1 
+
{2 \pi \lambda p d_0^2 e^{2\delta \over d_0} \over \sigma_{\dB}^4}
+
{6 \pi \lambda p d_0^2 e^{3\delta \over d_0} \over \sigma_{\dB}^6}
+
O\left({e^{4\delta \over d_0} \over \sigma_{\dB}^8}\right)
\right].
\label{eq-app-spatial-n=2}
\end{eqnarray}
\end{lem}
\begin{proof}
For simplicity, we write $\sigma \equiv \sigma_{\dB}$ in this proof. 
In addition, we only show the case of $n =2$ in what follows. 
For the complete proof, please see Appendix~\ref{proof-corr-spatial}. 

By applying Campbell's theorem and using (\ref{eq-lognorm-corr}) and (\ref{eq-def-U}), the second term in (\ref{eq-spa-corr-0}) can be calculated as
\begin{eqnarray}
\lefteqn{
\bbE \left[
\sum_{x_i\in \Phi} {h_i \tilde{h}_i \calS_i \over \ell(\|x_i\|)\ell(\|x_i - \tilde{o}_\delta\|) }\right]
} &&
\nonumber
\\
&=&
\lambda p \int_{\bbR^n} {e^{\sigma^2 e^{- \delta \over d_0}} \over \ell(\|x\|)\ell(\|x - \tilde{o}_\delta\|) } \rmd x 
=
\lambda p e^{\sigma^2 e^{- \delta \over d_0}} U_n(\delta).
\label{eq-spa-cor-term-1}
\end{eqnarray}
Thus, in what follows, we aim to calculate the first term in (\ref{eq-spa-corr-0}). Similar to the above, by using 
Campbell's theorem and (\ref{eq-assu-1}), we obtain
%
\begin{eqnarray}
\lefteqn{
\bbE
\left[
\sum_{x_i, x_j \in \Phi_{\neq}^{(2)}}
{ h_{i} \tilde{h}_{j} \calS(x_i) \calS(x_j) \over \ell(\|x_i\|) \ell(\|x_j - \tilde{o}_\delta\|)}
\right] 
} &&
\nonumber
\\
&=&
\lambda^2 p^2
\int \!\!\!
\int_{(\bbR^2)^2}
{ \exp\left(\sigma^2 e^{-{\|x - y\| + \delta \over d_0}}\right)
\over \ell(\|x + \tilde{o}_\delta\|) \ell(\|y\|)}
\rmd x \rmd y.
\label{eq-int-spa-cor-1}
\end{eqnarray}
%
%
For any fixed $x \in \bbR^2$, by considering a planar coordination (i.e., $x:= (r \cos \theta, r \sin \theta), y := (s \cos \phi, s \sin \phi)$), we obtain
%
\begin{eqnarray}
\lefteqn{
\int_{\bbR^2}\int_{\bbR^2}
{\exp\left(\sigma^2 e^{-{\|x - y\| + \delta \over d_0}}\right)
\over
\ell(\|x + \delta\|)
\ell(\|y\|)} \rmd y \rmd x
} &&
\nonumber
\\
&=&
\int_0^{2\pi} \!\!\! \int_0^{2\pi} \!\!\!
\int_0^\infty \!\!\! 
{r
\over
\ell(\sqrt{r^2 + \delta^2 - 2 r \delta \cos \theta})
} 
\nonumber
\\
&{}&
\times
\int_0^\infty \!\!\!
{s e^{\sigma^2 e^{-{s + \delta \over d_0}}} \rmd s \rmd r \rmd \phi \rmd \theta 
\over \ell(\sqrt{r^2 + s^2 + 2 sr \cos(\phi -\theta)})}.
\label{eq-spa-2-add-int-2}
\end{eqnarray}
Note here that (see (\ref{eq-def-U}))
\begin{eqnarray}
U_2(\delta) &=& 
\int_0^{2\pi}\!\!\!\int_{\bbR_+}{ r 
\over (\eps_0 + r^{\alpha}) \ell(\sqrt{r^2 + \delta^2 - 2 r \delta \cos \theta})}\rmd r \rmd \theta.
\nonumber
\end{eqnarray}
Note also that 
\begin{equation}
\int_0^{2\pi}
{r^{\alpha -1} \cos (\phi - \theta) \over (\eps_0 + r^\alpha )^2}
\rmd \phi 
=0.
\nonumber
\end{equation}
%
Thus, by applying Corollary~\ref{lem-watson} to (\ref{eq-spa-2-add-int-2}) and using the above lead to
\begin{eqnarray}
\lefteqn{
\int_{\bbR^2}\int_{\bbR^2}
{\exp\left(\sigma^2 e^{-{\|x - y\| + \delta \over d_0}}\right)
\over
\ell(\|x + \tilde{o}_\delta\|)
\ell(\|y\|)} \rmd y \rmd x
}&&
\nonumber
\\
&\sim&
d_0^2 e^{\sigma^2 e^{- \delta \over d_0}}
\int_0^{2\pi} \!\!\! \int_0^{2\pi} \!\!\!
\int_0^\infty \!\!\! 
{r \rmd \phi \rmd \theta \rmd r 
\over
\ell(\sqrt{r^2 + \delta^2 - 2 r \delta \cos \theta})
} 
\nonumber
\\
&{}&
\times
\left[{ 1 \over \eps_0 + r^\alpha}{e^{2\delta \over d_0} \over \sigma^4}
+ \left({3 \over \eps_0 + r^\alpha}
\right.
\right.
\nonumber
\\
&{}&
\left.
\left.
-
{2 \alpha d_0 r^{\alpha-1} \cos (\phi - \theta) \over (\eps_0 + r^\alpha)^{2}}
\right){e^{3\delta \over d_0} \over \sigma^6}
+ O\left({e^{4\delta \over d_0} \over (\eps_0 + r^{\alpha} )\sigma^8}\right)
\right]
\nonumber
\\
&=&
2 \pi d_0^2 e^{\sigma^2 e^{- \delta \over d_0}} U_2(\delta)
\left[
{e^{2\delta \over d_0} \over \sigma^4}
+
{3e^{3\delta \over d_0} \over \sigma^6}
+
O\left({e^{4\delta \over d_0} \over \sigma^8}\right)
\right],\qquad
\nonumber
\end{eqnarray}
%
Therefore, substituting (\ref{eq-spa-cor-term-1}) and the above into (\ref{eq-spa-corr-0}) results in (\ref{eq-app-spatial-n=2}). 
\end{proof}

We obtain asymptotic expansions of the spatial correlation coefficient by combining Lemma~\ref{thm-corr-spatial}, (\ref{eq-app-EI^2-n=1}), and (\ref{eq-app-EI^2-n=2}) with (\ref{eq-def-rho-sp}) and (\ref{eq-mean-order}). 
\begin{thm}\label{thm-spa-cor}
If $\delta < d_0 \log \sigma^2_{\dB}$, then the spatial correlation coefficients of interference for distance $\delta$ is asymptotically equivalent to as $\sigma_{\dB}^2 \to \infty$:i) if $n = 1$, 
\begin{eqnarray}
\lefteqn{
\rho_{\rsp, \delta} \sim
{
1
\over
\gamma_1
\left(
1 + {2 \lambda p d_0 \over \sigma^2_{\dB}} + {2 \lambda p d_0 \over \sigma^4_{\dB}}
\right)
+
{\lambda p d_0^2 \over \eps_0^2 \sigma^4_{\dB}}
}
} &&
\nonumber
\\
&\times&
\left[
e^{\sigma_{\dB}^2(e^{{-\delta \over d_0}} - 1)}
\left\{
U_1(\delta)
\left(
1+ {2 \lambda p d_0 e^{\delta \over d_0} \over \sigma^2_{\dB}} 
+ {2 \lambda p d_0 e^{2\delta \over d_0}\over \sigma^4_{\dB}}
\right)
\right.
\right.
\nonumber
\\
&+&
\left.\left.
{2 \lambda p d_0^2 e^{2\delta \over d_0} \over \eps_0(\eps_0 + \delta^\alpha) \sigma^4_{\dB}}
\right\}
+
e^{\sigma_{\dB}^2(e^{-{2\delta \over d_0}} - 1)}
{\lambda p d_0^2 e^{4 \delta \over d_0} \over \eps_0^{2} \sigma^4_{\dB}}
\right].\quad~~~
\label{eq-app-spatial-corr}
\end{eqnarray}
ii) if $n = 2$, 
%
\begin{equation}
\rho_{\rsp, \delta} 
\sim
{
e^{\sigma_{\dB}^2 (e^{ {-\delta \over d_0}} -1)}
U_2(\delta)
\over 
\gamma_2
}
\cdot
{
1 
+
{2 \pi \lambda p d_0^2 e^{2\delta \over d_0} \over \sigma_{\dB}^4}
+
{6 \pi \lambda p d_0^2 e^{3\delta \over d_0} \over \sigma_{\dB}^6}
\over
1 + {2\pi \lambda p d_0^2 \over \sigma_{\dB}^4} 
+ {6 \pi \lambda p d_0^2 \over \sigma_{\dB}^6}
}.\qquad
\label{eq-app-spatial-corr-n=2}
\end{equation}
\end{thm}

\begin{rem}\label{rem-app-spa}
The expressions in Theorem~\ref{thm-spa-cor} are still in intractable  forms due to the function $U_n(\delta)$ defined in (\ref{eq-def-U}). However, from the definition of $\ell(\cdot)$, this function can be upper-bounded as follows (see (\ref{eq-E[I^2]-sec-term}) and (\ref{eq-def-U})):
\begin{equation}
U_n(\delta) \le U_n(0) = \gamma_n.
\label{eq-upper-U(delta)}
\end{equation}
As shown later (see Section~\ref{subsubsec-accura}), by directly applying the above upper bound to (\ref{eq-app-spatial-corr}), we can obtain tight and closed-form approximate formulas for the spatial correlation coefficients of interference. In addition, we can reduce the expression of $\rho_{\rsp, \delta}$ in (\ref{eq-app-spatial-corr}) by omitting the terms of $O(1/\sigma^4_{\dB})$ (resp. $O(1/\sigma^6_{\dB})$) when $n = 1$ (resp. $n=2$) as follows:
%
\begin{eqnarray}
\rho_{\rsp, \delta}&\sim&
e^{\sigma^2_{\dB}(e^{-\delta \over d_0} -1)}
{U_n(\delta) \over \gamma_n}
{1 + { \lambda p n V_n d_0^n e^{\delta \over d_0}\over \sigma^{2n}_{\dB}}
\over
1 + {\lambda p n V_n d_0^n\over \sigma^{2n}_{\dB}}
}
\nonumber
\\
&\le&
e^{\sigma^2_{\dB}(e^{-\delta \over d_0} -1)}
{\sigma^{2n}_{\dB} + \lambda p n V_n d_0^ne^{\delta \over d_0}
 \over 
\sigma^{2n}_{\dB}+ \lambda p n V_n d_0^n
}
. 
\label{eq-spa-cor-simp}
\end{eqnarray}
%
The above result can be used as a simpler version of approximate formulas for the spatial correlation coefficients of interference. 
\end{rem}

\subsection{Temporal Correlation of Interference}
\label{sec-temp}
In this section, we analyze the temporal correlation of interference, i.e., the correlation between the interferences at two different time slots. Due to the correlated shadowing, the interference is expected to be more correlated in a lower mobility environment. We first derive an asymptotic expansion of the temporal correlation coefficient of interference that is valid for general i.i.d. mobility models. We then consider several commonly used mobility models as examples and derive the corresponding spatial correlation coefficient for each model. 

Recall first that the moving vector $v_i(t)$ of the $i$-th node during one time-slot is independently distributed in accordance with the p.d.f. $\psi(\|v_i(t)\|)$ and $\psi(\cdot)$ is rotation invariant. Therefore, the moving distance of the transmitter $i$ in $\tau$ time-slots, i.e., $\overline{v}_i(\tau) \triangleq \|x_i(t + \tau) - x_i(t)\|$, has the p.d.f. $\psi_\tau(v) \triangleq \psi^{\tau\ast}(v)$ ($v \in \bbR_+$), which represents the $\tau$-th convolution of $\psi(v)$ and does not depend on $x_i(t)$. 
The following results are asymptotic expressions of the temporal covariance and correlation coefficient of interference. 
\begin{lem}\label{thm-cov-temp}
Suppose that $\psi_\tau(0) > 0$ and $\psi'_\tau(0)$ exists. The temporal covariance of interference for time interval $\tau$ has the following asymptotic expansions as $\sigma_{\dB}^2 \to \infty$: i) if $n = 1$, 
\begin{eqnarray}
&{}&\Cov[I(t),I(t+\tau)] \sim
2 \lambda p^2 d_0 e^{\sigma_{\dB}^2} 
\left[
\gamma_1
\left\{
{(\lambda + \psi_\tau(0)) \over \sigma_{\dB}^2} 
\right.
\right.
\nonumber
\\
&{}&
\left.
\left.
+
{\lambda + \psi_\tau(0) + d_0 \psi'_\tau(0)\over \sigma_{\dB}^4}
\right\}
+ 
{ d_0 (\lambda + \psi_\tau(0)) \over \eps^2_0 \sigma_{\dB}^4}
+
O\left({1 \over \sigma_{\dB}^6}\right)
\right],
\nonumber
\\
\label{eq-spa-temp-cov}
\end{eqnarray}
ii) if $n = 2$, 
\begin{eqnarray}
&{}&\Cov[I(t),I(t+\tau)] \sim
2 \pi \lambda p^2 d_0^2 e^{\sigma_{\dB}^2} \gamma_2
\left[
{\lambda + \psi_\tau(0) \over \sigma_{\dB}^4}
\right.
\nonumber
\\
&{}&
\left.
+
{(3 \lambda + 3 \psi_\tau(0) + 2 d_0 \psi'_\tau(0)) \over \sigma_{\dB}^6} 
+
O\left({1 \over \sigma_{\dB}^8}\right)
\right],
\label{eq-spa-temp-cov-n=2}
\end{eqnarray}
where $\gamma_n$ ($n=1,2$) are given in (\ref{eq-E[I^2]-sec-term}). 
\end{lem}
\begin{proof}
By definition, $\bbE[I(t)I(t + \tau)]$ can be expressed as
\begin{eqnarray}
\lefteqn{
\bbE[I(t) I(t+\tau)]
}
\nonumber
\\
&=&
\bbE \left[\sum_{x_i(t), x_j(t) \in \Phi^{(2)}_{\neq}(t)} {h_i(t) h_j(t+\tau) \calS_i(t)\calS_j(t+\tau) 
\over \ell(x_i(t)) \ell(x_j(t + \tau))}
\right]
\nonumber
\\
&{}&
+
\bbE\left[
\sum_{x_i(t)\in \Phi(t)} {h_i(t) h_i(t+\tau) \calS_i(t)\calS_i(t+\tau) 
\over \ell(x_i(t))\ell(x_i(t+\tau)) }\right].
\label{eq-temp-corr-0}
\end{eqnarray}
Recall that the shadowing effect is assumed to be time-invariant, i.e., 
the value of the shadowing variable at a fixed node location does not change over time.
It thus follows from Campbell's theorem that, for $n =1,2$,
\begin{eqnarray}
\lefteqn{
\bbE\left[
\sum_{x_i(t), x_j(t) \in \Phi_{\neq}^{(2)}(t)} \!\!\!\!
{ h_{i}(t) h_{j}(t+ \tau )\calS_i(t)\calS_j(t+\tau) \over \ell(x_i(t)) \ell(x_j(t+\tau))}
\right]
}\qquad\quad
&&
\nonumber
\\
&=&
\lambda^2 p^2
\int_{\bbR^n}\int_{\bbR^n}
\bbE_v \left[
{e^{\sigma_{\dB}^2 e^{ - {\|x- (y + v)\| \over d_0}}} \over \ell(\|x\|) \ell(\|y + v\|)}
\right]
\rmd x \rmd y
\nonumber
\\
&=&
\lambda^2 p^2
\int_{\bbR^n} \int_{\bbR^n} \int_{\bbR^n}
{e^{\sigma_{\dB}^2 e^{ - {\|s - v\| \over d_0}}}\psi_\tau(\|v\|) \over \ell(\|x\|) \ell(\|x - (s - v)\|)}
\rmd v
\rmd x \rmd s
\nonumber
\\
&=&
\lambda^2 p^2
\int_{\bbR^n} \int_{\bbR^n}
{e^{\sigma_{\dB}^2 e^{ - {\|s\| \over d_0}}} \over \ell(\|x\|) \ell(\|x + s\|)}
\rmd x \rmd s,
\label{ep-spa-temp-cov-1}
\end{eqnarray}
which shows that the first term in (\ref{eq-temp-corr-0}) is insensitive to the nodes mobility. In addition, we can see the above integral can be calculated by Proposition~\ref{lem-appen-asym}. Furthermore, applying Campbell's theorem to the second term in (\ref{eq-temp-corr-0}) yields 
\begin{eqnarray}
\lefteqn{
\bbE\left[
\sum_{x_i(t)\in \Phi(t)}
{ h_{i}(t) h_{i}(t+\tau) \calS_i(t)\calS_i(t+\tau))\over \ell(x_i(t)) \ell(x_i(t+\tau))}
\right]
}
&&
\nonumber
\\
&=&
\lambda p^2 \int_{\bbR^n}
\bbE_v\left[
{ e^{\sigma^2_{\dB} e^{- {\|v\| \over d_0}}} \over \ell(\|x\|) \ell(\|x + v\|)}
\right]
\rmd x 
\nonumber
\\
&=&
\lambda p^2 \int\!\!\!\int_{(\bbR^n)^2}\! 
{ e^{\sigma^2_{\dB} e^{- {\|v\| \over d_0}}} \psi_{\tau}(\|v\|) \over \ell(\|x\|) \ell(\|x+ v\|)}
\rmd v
\rmd x
\sim \lambda p^2 \Psi_n (\psi_\tau).\quad
\label{ep-spa-temp-cov-2}
\end{eqnarray}
where the last approximation is due to Proposition~\ref{lem-appen-asym}. Therefore, applying Proposition~\ref{lem-appen-asym} to (\ref{ep-spa-temp-cov-1}) and combining this with (\ref{eq-def-cov-temp}), (\ref{eq-spa-corr-0}), and (\ref{ep-spa-temp-cov-2}), we have (\ref{eq-spa-temp-cov}). 
\end{proof}
By combining Lemmas~\ref{thm-second-moment} and \ref{thm-cov-temp} with (\ref{eq-def-rho_tm}), 
we can readily obtain the following theorem.  
\begin{thm}\label{thm-corr-temp}
Suppose that $\psi_\tau(0) > 0$ and $\psi'_\tau(0)$ exists. The temporal correlation coefficients of interference of time interval $\tau$ is asymptotically equivalent to, as $\sigma_{\dB} \to \infty$: i) if $n = 1$, 
\begin{eqnarray}
\lefteqn{
\rho_{\rtm, \tau} \sim
{
2 p d_0 
\over
1 + {2 \lambda p d_0\over \sigma^2_{\dB}} + {2 \lambda p d_0 \over \sigma^4_{\dB}}
+
{2d_0^2 \lambda p \over \gamma_1 \eps_0^2 \sigma^4_{\dB}}
}
\left[
{\lambda + \psi_\tau(0) \over \sigma_{\dB}^2} 
\right.
} &&
\nonumber
\\
&+&
\left.
{\lambda + \psi_\tau(0) + d_0 \psi'_\tau(0) \over \sigma_{\dB}^4}
+{d_0 (\lambda + \psi_\tau(0)) \over \gamma_1 \eps_0^2 \sigma_{\dB}^4}
\right]
,\qquad
\nonumber
\end{eqnarray}
ii) if $n = 2$, 
\begin{eqnarray}
\lefteqn{
\rho_{\rtm, \tau} \sim
{
2 \pi p d^2_0 
\over
1 + 
{2\pi \lambda p d^2_0 \over \sigma_{\dB}^4}
+
{6 \pi \lambda p d^2_0 \over \sigma_{\dB}^6}
}
} &&
\nonumber
\\
&\times&
\left[
{\lambda + \psi_\tau(0) \over \sigma_{\dB}^4}
+
{(3 \lambda + 3 \psi_\tau(0) + d_0 \psi'_\tau(0)) \over \sigma_{\dB}^6} 
\right]
.
\nonumber
\end{eqnarray}
\end{thm}

Theorem~\ref{thm-corr-temp} indicates that $\rho_{\rtm, \tau}$ depends only on $\psi_\tau(0)$ and $\psi'_\tau(0)$ and does not depend on $\psi_\tau(v)$ for $v \neq 0$ when $\sigma_{\dB} \to \infty$. In other words, the probability that the transmitters do not move has the dominant impact on the temporal correlation coefficient of interference when $\sigma_{\dB}$ is large. 

\begin{rem}
Similar to (\ref{eq-spa-cor-simp}), by removing the terms of $O(1/\sigma^4_{\dB})$ or $O(1/\sigma^6_{\dB})$, we can obtain simpler asymptotic expressions of $\rho_{\rtm, \tau}$ for each $n = 1, 2$, as follows.
%
%
\begin{eqnarray}
\rho_{\rtm, \tau} \sim
\dm{
p n V_n d_0^n (\lambda + \psi_\tau(0))
\over
\lambda p n V_n d_0^n + \sigma_{\dB}^{2n}
},
\label{eq-temp-cor-simple}
\end{eqnarray}
%
\end{rem}
%
%
%
%
%
%
%
\begin{rem}\label{rem-tmp-high-mobility}
The asymptotic expansions in (\ref{eq-temp-cor-simple}) show that 
\begin{eqnarray}
\rho_{\rtm, \tau} \sim
\dm{
\lambda p n V_n d_0^n 
\over
\lambda p n V_n d_0^n + \sigma_{\dB}^{2n}
},
\quad \mbox{as $\psi_\tau(0) \to 0$}, 
\nonumber
\end{eqnarray}
which indicates that the temporal correlation of interference does not disappear even in a very-high mobility environment, i.e., $\psi_\tau(0) \to 0$. 
On the other hand, if no spatial correlation of shadowing exists,
the temporal correlation of interference slowly decreases to 0 as $\psi_\tau(0) \to 0$~\cite{Gong14}. 
This surprising difference can be primarily attributed to 
the first term in  (\ref{eq-temp-corr-0}), i.e., 
the cross correlation of interference 
from a node ($x_i(t)$) at time slot $t$ 
and those from other nodes ($x_j(t + \tau)$ $(i \neq j)$) at time slot $t + \tau$, 
which is determined by the distance $\|x_i(t) - x_j(t + \tau)\|$ for each node. 
To explain this more precisely, we denote
the first and second terms in  (\ref{eq-temp-corr-0}) by $Q_1$ and $Q_2$, respectively. 
Owing to the random i.i.d. mobility of nodes, 
$Q_1$ is, on average, equivalent to the cross correlation of interference 
from $x_i(t)$ and $x_j(t)$ $(i \neq j)$
at the time slot $t$ (see (\ref{ep-spa-temp-cov-1})), i.e., 
\begin{equation}
Q_1 
= 
\bbE \left[\sum_{x_i(t), x_j(t) \in \Phi^{(2)}_{\neq}(t)} {h_i(t) h_j(t) 
\calS_i(t)\calS_j(t) 
\over \ell(x_i(t)) \ell(x_j(t))}
\right].
\label{eq-ref-Q_1}
\end{equation}
Thus, $Q_1$ does not depend on the time interval $\tau$
or on node mobility. 
%
Note that the relationship (\ref{eq-ref-Q_1}) holds even in the 
case of independent shadowing~\cite{Gong14}. 
However, if the shadowing is not correlated, 
the calculation of the right-hand side of 
(\ref{eq-ref-Q_1}) is much simplified 
in comparison with the case of correlated shadowing. 
More precisely, since 
$\bbE[h_i(t)h_j(t +\tau)] = \bbE[h_i(t)]\bbE[h_j(t +\tau)]$
in the case of independent shadowing, (\ref{eq-ref-Q_1}) yields 
%
\begin{align}
Q_1&=
\lambda^2 p^2 (\bbE[h])^2
\int_{\bbR^n}\int_{\bbR^n}
\bbE_v \left[
{1 \over \ell(\|x\|) \ell(\|y + v\|)}
\right]
\rmd x \rmd y
\nonumber
\\
&~~
=
\lambda^2 p^2 
\left(\int_{\bbR^n}{\rmd x \over \ell(\|x\|)} \right)^2
= (\bbE[I])^2,
\label{eq-non-spa}
\end{align}
where the last equality follows from (\ref{eq-mean-I}). 
Thus, by combining (\ref{eq-non-spa}) with (\ref{eq-def-cov-temp}), 
we obtain $\Cov[I(t)I(t+\tau)] = Q_2$. 
Moreover, $Q_2$ represents the mean product of 
the interferences from a node at different time slots
(i.e., $x_i(t)$ and $x_i(t+\tau)$), and thus
it decreases to zero when $\psi_\tau(0)\to 0$. 
On the other hand, if the shadowing is correlated, 
i.e., $\bbE[h_i(t)h_j(t +\tau)] \neq \bbE[h_i(t)]\bbE[h_j(t +\tau)]$,
then $Q_1\neq (\bbE[I])^2$. 
As a result, the temporal covariance approaches
$Q_1 - (\bbE[I])^2 \neq 0$ as $\psi_\tau(0) \to 0$
and does not disappear due to correlated shadowing. 
\end{rem}

\begin{rem}
Contrary to the temporal correlation of interference, 
the spatial correlation of interference decreases to zero
when the distance $\delta$ between the receivers
increases. This fact is due to the 
difference between the scenarios considered in 
the spatial and temporal correlation analysis
and Assumption~\ref{assu-corr}. 
More precisely, we assume that 
in the spatial correlation analysis, 
only the positions of the receivers are different (i.e., $o$ and $\tilde{o}_\delta$)
and all the other nodes are fixed. 
In contrast, in the temporal correlation analysis, 
the position of the receiver is assumed to be fixed over time while the other nodes move randomly.
%
Moreover, due to Assumption~\ref{assu-corr}, 
in the spatial correlation analysis, 
the correlation coefficient of 
$h_i$ and $\tilde{h}_j$ is determined
by the sum of $\delta$ and the distance between nodes $x_i$ with $x_j$
 (see (\ref{eq-assu-1})). 
Thus, as $\delta$ increases, 
$\bbE[h_i \tilde{h}_j]$ decreases to zero for 
all $i,j \in \bbZ_+$. As a result, 
the spatial covariance $\bbE[I_0\tilde{I}_\delta]$ decreases to zero
as $\delta$ increases. 
On the other hand, the mean product $\bbE[h_i(t)h_j(t + \tau)]$
in the temporal correlation analysis depends only on the distance
between $x_i(t)$ and $x_j(t+\tau)$. Thus, on average, the impact of the 
mobility of nodes disappears, as shown in (\ref{ep-spa-temp-cov-1}), 
and the temporal covariance does not approach zero in a high
mobility environment (see also Remark~\ref{rem-tmp-high-mobility}).  
\end{rem}

\subsubsection{Mobility models}
\label{subsec-exam-mobil}
We next consider several commonly used mobility models as examples and show the temporal correlation coefficients corresponding to each model. In this paper, we choose three models: (i) {\it constrained i.i.d. mobility (CIM) model}; (ii) {\it random walk (RW) model}; and (iii) {\it discrete-time Brownian motion (BM) model}. The same models are considered by Gong and Haenggi~\cite{Gong14}. In what follows, we describe the details of each model. 

\medskip

\noindent
(i) {\it constrained i.i.d. mobility (CIM) model}: In the CIM model, the location $x_i(t + 1) \in \Phi(t+1)$ of the transmitter $i$ at the time slot $(t+1)$ is determined independently of $x_i(t)$ such that
\[
x_i(t + 1) := x_i(0) + v_{i}(t), \qquad t \in \bbZ_+.
\]
Here, $v_i(t)$'s are uniformly distributed within $B(o, R_{\mathrm{CIM}})$, where $B(x, r)$ denotes a open circle centered at $x$ with radius $r$. Therefore, under the CIM model, the p.d.f. $\psi_\tau(v)$ of the moving distance in $\tau$ time slots can be represented as
\[
\psi_{\tau}(v) = \left\{ 
\begin{array}{ll}
{1 \over |B(o, R_{\mathrm{CIM}})|}, & v \in B(o, R_{\mathrm{CIM}}),\\
0, & \mbox{otherwise}.
\end{array}
\right.
\]
\smallskip

\noindent
(ii) {\it random walk (RW) model}: In the RW model, the location of the transmitter $i$ at time $(t + 1)$ is determined as follows:
\begin{equation}
x_i(t + 1) := x_i(t) + v_{i}(t), \qquad t \in \bbZ_+. 
\label{eq-mobility}
\end{equation}
Similar to in the CIM model, $v_{i}(t)$'s are distributed within $B(o, R_{\mathrm{RW}})$. Note that the RW model differs from the CIM model because the location of the transmitter $i$ at time $(t+1)$ depends on that at time $t$. Under the RW model, $\psi_\tau(v)$ becomes the $\tau$-th convolution of a uniform distribution. Specifically, if $n = 1$, $\psi_\tau(0)$ becomes (see e.g., \cite{Reny70})
\[
\psi_\tau(0) = {1 \over 2 R_{\mathrm{RW}}}\sum_{i=0}^{\lfloor {\tau \over 2} \rfloor}
(-1)^i {\tau \over \tau!(\tau -i)!} \left({\tau \over 2} - i\right)^{\tau - 1}
\equiv \bar{C}_{\mathrm{RW}},
\]
because $\psi(v)$ is a uniform distribution with range $[-R_{\mathrm{RW}}, R_{\mathrm{RW}}]$. In this case, $\psi_\tau(v)$ is not differentiable at $v = 0$ when $\tau$ is an even number. 

\smallskip
\noindent
(iii) {\it discrete-time Brownian motion model (BM)}: In the random walk model, the location of the transmitter $i$ at time $t + 1$ is given by (\ref{eq-mobility}) and each element of $v_{i}(t)$ is distributed with $\calN(0, \sigma_V^2)$. In this model, the p.d.f. $\psi_\tau(v)$ simply equals $\calN(0, \tau \sigma_V^2)$ due to a property of a normal distribution. 

\medskip
\begin{rem}
Note that in the above three models, the distribution of the locations of the nodes at time $t$ becomes again a homogeneous PPP with intensity $\lambda$ due to the displacement property of PPPs (see e.g., \cite{Haen13}). 
\end{rem}

We can easily confirm that $\psi_\tau'(0)= 0$ in the above three models. As a result, we obtain the following results. 

\begin{coro}
Consider the CIM, RW, and BM models described in Section~\ref{subsec-exam-mobil}. The temporal correlation coefficient corresponding to each model has the following form of an asymptotic expansion as $\sigma_{\dB} \to \infty$: if $n = 1$, 
\begin{eqnarray}
\rho_{\rtm, \tau} &\sim&
{
2 (\lambda + C^{(1)}_\tau )p d_0
\over
1 + {2 \lambda p d_0\over \sigma^2_{\dB}} + {2 \lambda p d_0 \over \sigma^4_{\dB}}
+
{d_0 \lambda \over \eps_0^2 \gamma_1\sigma^4_{\dB}}
}
\nonumber
\\
&{}&
\times
\left[
{1\over \sigma_{\dB}^2} 
+
{1\over \sigma_{\dB}^4}
+
{d_0 \over \eps_0^2 \gamma_1 \sigma_{\dB}^4}
\right],\qquad
\nonumber
\end{eqnarray}
and if $n = 2$, 
\begin{eqnarray}
\rho_{\rtm, \tau} \sim
{
(\lambda + C^{(2)}_\tau) 2 \pi p d_0^2 
\over
1 + 
{2\pi \lambda p d^2_0 \over \sigma_{\dB}^4}
+
{6 \pi \lambda p d^2_0 \over \sigma_{\dB}^6}
}
\left[
{1 \over \sigma_{\dB}^4}
+
{3 \over \sigma_{\dB}^6} 
\right].
\nonumber 
\end{eqnarray}
where $C_\tau^{(n)}$ ($n =1,2$) is given by (i) under the CIM model, 
\[
C^{(1)}_\tau = {2 \over R_{\mathrm{CIM}}},
\quad
C^{(2)}_\tau = {1 \over \pi R^2_{\mathrm{CIM}}};
\]
(ii) under the RM model, if $\tau$ is an odd number, 
\[
C^{(1)}_\tau = \bar{C}_{\mathrm{RW}},
\quad
C^{(2)}_{1} = {1 \over \pi R^2_{\mathrm{CIM}}},
\]
and (iii) under the BM model, 
\[
C^{(1)}_{\tau} = {1 \over \sqrt{2 \pi \tau \sigma_V^2}},
\quad
C^{(2)}_{\tau} = {1 \over 2 \pi \tau \sigma_V^2}.
\]
\end{coro}

\section{Numerical Examples}\label{sec-num}
We next present several numerical examples for the results obtained in the previous section. In all numerical results in this paper, we basically compared the spatial or temporal correlation coefficients calculated by two methods: i) calculating original integrals in the exact expressions of correlation coefficients by using a numerical integration method (Monte-Carlo integration method), and ii) approximating correlation coefficients on the basis of our asymptotic expansions. By doing this, we show the relationship between various system parameters with the spatial or temporal correlation coefficients of interference and usefulness of the asymptotic expansions as closed-form approximate formulas. 
\subsection{Spatial Correlation of Interference}
\label{subsec-spa-num}
\subsubsection{Impacts of parameters}
We first show the results for the spatial correlation of interference. 
Fig.~\ref{fig:rho-spa-delta} show the results for $\rho_{\rsp, \delta}$ with different $\sigma_{\dB}$ when varying the distance $\delta$ between the receivers at $o$ and $\tilde{o}_\delta$. The upper graph represents the case of $n = 1$ and the lower $n = 2$. Other parameters are set as $d_{\mathrm{cor}} = 0.1$~[km] $\alpha = 4$, $p = 1$, and $\eps_0 = 0.001$. The points with the label ``num" correspond to the results from the numerical integration method, and the lines with the label ``app" are from the approximation method based on Theorem~\ref{thm-spa-cor}. We found that the spatial correlation rapidly decreases as $\delta$ increases. We also found that the distance at which $\rho_{\rsp, \delta} = 0.5$, i.e., the correlation distance of interference, is much smaller than that of shadowing $d_{\mathrm{cor}}$. This means that the spatial correlation of shadowing does not affect the spatial correlation of interference in the same scale. In addition, the graphs show that if $\sigma_{\dB}$ increases, the spatial correlation coefficients rapidly decrease. 
%
Since a realistic value of $\sigma_{\dB}$ is between 3 and 15 dB~\cite{Gold05}, the spatial correlation almost disappears when $\sigma_{\dB}$ is high. 
We can also see from the graphs that in all cases, the approximate values well fit those from the numerical integration method.

In Fig.~\ref{fig:rho-spa-dzero}, we plot $\rho_{\rsp, \delta}$ with different $d_{\mathrm{cor}}$. We fixed $\delta$ to $0.01$ and used the same parameters as in Fig.~\ref{fig:rho-spa-delta}. Since $d_{\mathrm{cor}}$ can be considered as the strength of the spatial correlation of shadowing, $\rho_{\rsp, \delta}$ increases as $d_{\mathrm{cor}}$ increases. Furthermore, we can see that the approximate values well fit the results from the numerical integration method. Therefore, the asymptotic formulas are useful for understanding the relationship between the spatial correlation of interference and the correlation distance of shadowing. 

\begin{figure}[!t]
\centering
\includegraphics[width=2.5in]{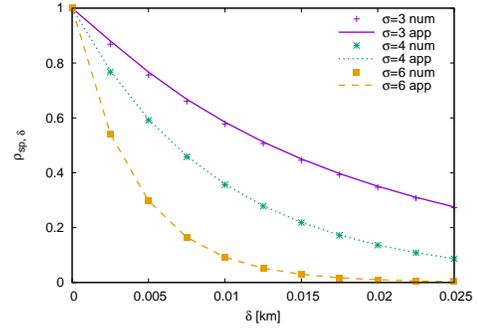}
\begin{center}
\scriptsize{(a) $n = 1$, $\lambda = 60$.}
\end{center}
\includegraphics[width=2.5in]{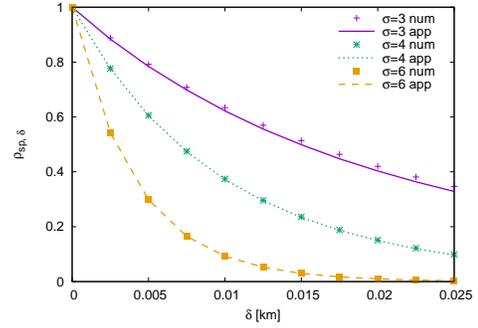}
\begin{center}
\scriptsize{(b) $n = 2$, $\lambda = 2000$.}
\end{center}
\vspace{-5.0mm}
\caption{Comparison of $\rho_{\rsp,\delta}$ from numerical integration and approximation with different $\sigma_{\dB}$ when varying $\delta$ [km].}
\label{fig:rho-spa-delta}
\vskip -7pt
\end{figure}
\begin{figure}[!t]
\centering
\includegraphics[width=2.5in]{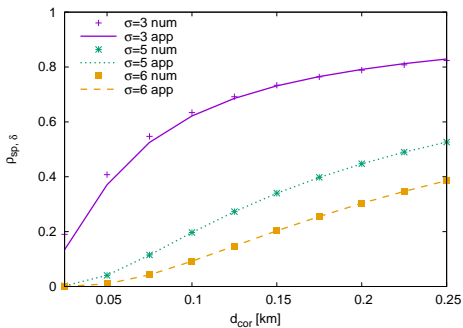}
\caption{Comparison of $\rho_{\rsp,\delta}$ from numerical integration and approximation with different $\sigma_{\dB}$ when varying $d_{\mathrm{cor}}$. $n= 2$ and $\lambda = 2000$.}
\label{fig:rho-spa-dzero}
\vskip -7pt
\end{figure}

\subsubsection{Accuracy of simplified approximate formulas}
\label{subsubsec-accura}
As shown in the above numerical examples, the approximate formulas for $\rho_{\rsp, \delta}$ based on the asymptotic expansions in Theorem~\ref{thm-spa-cor} well fitted the results from numerical integration methods. However, the approximate formulas contain an intractable  function $U_n(\delta)$, and thus we presented simpler versions of the approximate formulas in Remark~\ref{rem-app-spa}. Therefore, we next illustrate their accuracy as approximate formulas. More precisely, we compare the approximate values based on (\ref{eq-app-spatial-corr-n=2}) with those based on (\ref{eq-spa-cor-simp}). Fig.~\ref{fig:rho-spa-app} compares the results for the original (shown as ``app") and simpler (shown as ``app2") versions of approximate formulas for $\rho_{\rsp, \delta}$ when varying $\sigma_{\dB}$. The other parameters were the same as those in Fig.~\ref{fig:rho-spa-delta}. 
As we can observe from the graph, although ``app" and ``app2" have 
slightly different values with the results of the numerical integration method when $\sigma_{\dB} < 3$, 
both approximate formulas achieve high accuracy
if $\sigma_{\dB} \ge 3$. 
Since $\sigma_{\dB}$ typically falls within 3--15 dB~\cite{Gold05}, 
our approximate formulas are valid in realistic settings. 
In addition, as expected, if $\sigma_{\dB}$ increases, the accuracy of the approximate formulas also increases. 
%
\begin{figure}[!t]
\centering
\includegraphics[width=2.5in]{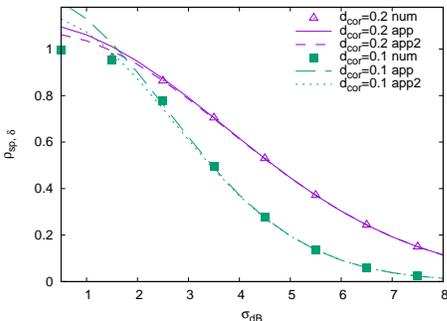}
\vspace{-5.0mm}
\caption{
Comparison of approximated formulas for $\rho_{\rsp,\delta}$ when varying $\sigma_{\dB}$. ``app" and ``app2" correspond to original and simplified versions of approximate formulas. $n = 2$, $\lambda = 2000$.}
\label{fig:rho-spa-app}
\vskip -7pt
\end{figure}
%


%

%
\begin{figure}[!t]
\centering
\includegraphics[width=2.5in]{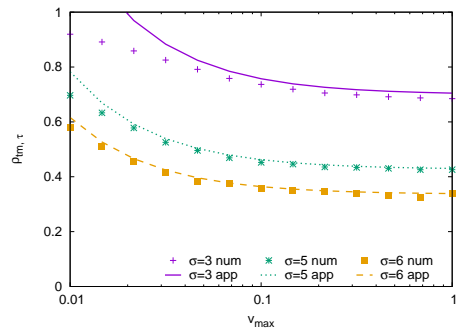}
\begin{center}
\scriptsize{(a) $n = 1$, $\lambda = 60$.}
\end{center}
\includegraphics[width=2.5in]{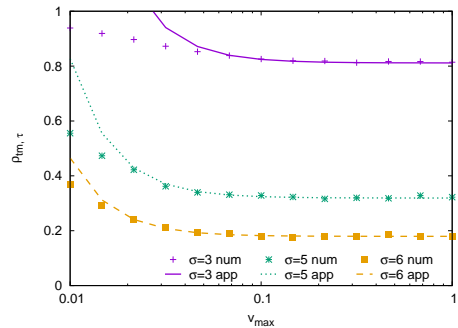}
\begin{center}
\scriptsize{(b) $n = 2$, $\lambda = 2000$.}
\end{center}
\vspace{-5.0mm}
\caption{Comparison of $\rho_{\rtm,1}$ from numerical integration and approximation under CIM model with different $\sigma_{\dB}$ when varying $v_{\max}$.}
\label{fig:rho-tmp-tau}
\vskip -7pt
\end{figure}
\begin{figure}[!t]
\centering
\includegraphics[width=2.5in]{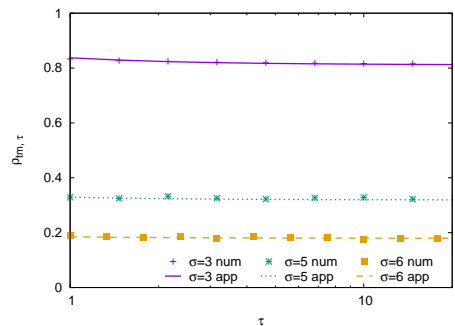}
\caption{Comparison of $\rho_{\rtm,\tau}$ from numerical integration and approximation under BRM model with different $\sigma_{\dB}$ when varying $\tau$. $\sigma^2_V = 0.0025$, $n =1$ and $\lambda = 60$.}
\label{fig:rho-tmp-tau-brm}
\vskip -7pt
\end{figure}

\subsection{Temporal Correlation of Interference}
\label{subsec-num-tmp}
We next provide numerical examples for the temporal correlation of interference. Similar to Section~\ref{subsec-spa-num}, we compared the exact values of the correlation coefficient computed by a numerical integration method with their approximate values on the basis of the asymptotic expansions in Theorem~\ref{thm-corr-temp}. 

Fig.~\ref{fig:rho-tmp-tau} compares the results of $\rho_{\rtm, \tau}$ from the numerical integration method (shown as ``num") and the approximation method (shown as ``app") under CIM mobility model with different $v_{\max}$. We set $\tau = 1$, $d_{\mathrm{cor}} = 0.1 $ [km], $\alpha = 4$, $p = 1$, and $\eps_0 = 0.001$. We can see that, if $\sigma_{\dB}$ increases, the temporal correlation of interference decreases similar to the spatial correlation. In addition, as we mentioned in Remark~\ref{rem-tmp-high-mobility}, if $v_{\max}$ increases, i.e., the mobility of nodes increases, the temporal correlation decreases but does not converge to 0. The graphs also show that if $\sigma_{\dB}$ is small, gaps between the approximate and exact values increase. This is because if the range of $\psi_\tau (v)$ is small, then Watson's lemma does not work well in our approximation. In short, this method assumes that the contribution from an exponential term in an integral function of the form (\ref{eq-tilde-F}) rapidly decreases from its maximum. However, if the integration range is too small, the impact from the exponential correlation function $\exp(\sigma^2 e^{-x / d_0})$ does not disappear within the range, and thus the approximation is likely to fail. (For details, see also (\ref{ep-spa-temp-cov-2}) and Section~\ref{subsec-pre}). 

Fig.~\ref{fig:rho-tmp-tau-brm} shows the results for $\rho_{\rtm, \tau}$ under a BRM model when varying $\tau$. We set $\sigma^2_V = 0.0025$, $n = 2$ and $\lambda = 2000$, and the other parameters are the same as those in Fig.~\ref{fig:rho-tmp-tau}. We can see that the temporal correlation of interference decreases slowly as $\tau$ increases but does not reach 0. A similar tendency was found in Fig.~\ref{fig:rho-tmp-tau}, which were also explained in Remark~\ref{rem-tmp-high-mobility}. Since we found that the difference in mobility models barely affects the behaviors of the temporal correlation coefficient of interference, we only show the case of the CIM model in what follows. 

Fig.~\ref{fig:rho-tmp-dzero} compares results for $\rho_{\rtm, 1}$ under CIM model with different $d_{\mathrm{cor}}$. The adopted parameters were the same as those in Fig.~\ref{fig:rho-tmp-tau}~(b). Similar to the spatial correlation of interference, if the correlation distance of shadowing increases, then the temporal correlation of interference increases (see also Fig.~\ref{fig:rho-spa-dzero}). In addition, if $d_{\mathrm{cor}}$ increases, then the gaps between the results from numerical and approximate methods also increase. The reason for this is similar to that in Fig.~\ref{fig:rho-tmp-tau} because if $d_{\mathrm{cor}}$ increases, the exponential correlation function tends to decrease slowly around $v = 0$ and thus Watson's lemma does not work well. 

Furthermore, we show the impact of $\lambda$ in Fig.~\ref{fig:rho-tmp-lam}. We used the same parameters as those in Fig.~\ref{fig:rho-tmp-dzero}. As we can see from the figure, if the density of nodes increases, the temporal correlation of interference increases. However, if the variance of the shadowing becomes larger, the increase becomes smaller. The reason for this can be considered as follows. As we explained in Remark~\ref{rem-tmp-high-mobility}, the cross correlations between interference from a node and those from other nodes at different time slots, i.e., $x_i(t)$ and $x_j(t + \tau)$'s $(i \neq j)$, do not depend on the node mobility. On the other hand, the cross correlation of interferences from a node at different time-slots, i.e., $x_i(t)$ and $x_i(t + \tau)$, depends on the node mobility and can reduce the temporal correlation. Indeed, we can see from (\ref{ep-spa-temp-cov-2}) and Proposition~\ref{lem-appen-asym} that this term decreases as $\psi_\tau(0)$ decreases, i.e., the node mobility becomes higher. Thus, as the density of nodes becomes larger, the contribution from the former becomes dominant in the temporal correlation coefficient of interference, i.e., the impact of the node mobility decreases. As a result, the overall temporal correlation increases. 

\begin{figure}[t]
\centering
\includegraphics[width=2.5in]{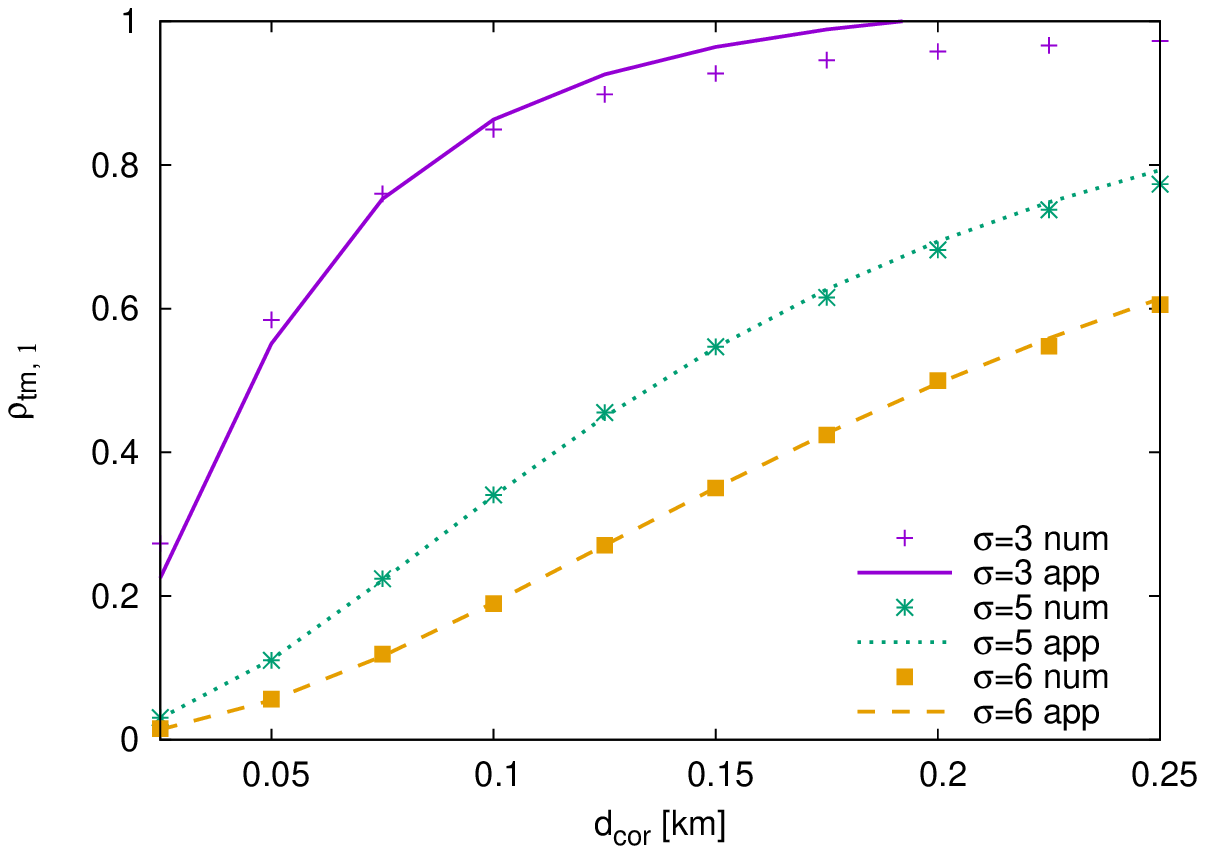}
\caption{Comparison of $\rho_{\rtm,1}$ from numerical integration and approximation under CIM model with different $\sigma_{\dB}$ when varying $d_{\mathrm{cor}}$. $v_{\max} = 0.05$, $n =2$ and $\lambda = 2000$.}
\label{fig:rho-tmp-dzero}
\vskip -7pt
\end{figure}
\begin{figure}[t]
\centering
\includegraphics[width=2.5in]{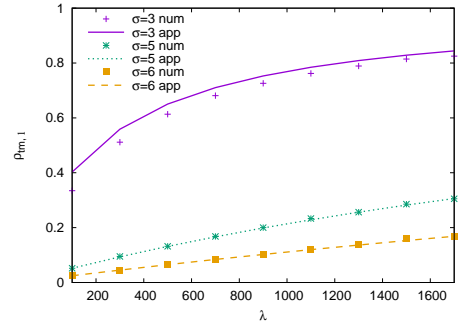}
\caption{Comparison of $\rho_{\rtm,1}$ from numerical integration and approximation under CIM model with different $\sigma_{\dB}$ when varying $\lambda$. $v_{\max} = 0.05$, $n =2$.}
\label{fig:rho-tmp-lam}
\vskip -7pt
\end{figure}
Finally, we investigate the difference between the cases of $n = 1$ and $n=2$. Fig.~\ref{fig:rho-tmp-sigma} shows the results for $\rho_{\rsp, \delta}$ under the CIM model with different $n$ and $d_{\mathrm{cor}}$. We set $\lambda = 60$ when $n = 1$ and $\lambda = 60^2$ when $n = 2$ so that $\lambda^n$ are equal. The figure shows that the decay rates of the temporal correlation of interference are different for $n = 1, 2$. This fact can be easily confirmed from Theorem~\ref{thm-corr-temp}, in which the case of $n = 1$ includes the term of $O(\sigma_{\dB}^2)$ when the case of $n = 2$ has the term of $O(\sigma_{\dB}^4)$. This fact indicates that the shape of a network is also important when considering a spatially correlated shadowing environment. 
Furthermore, the errors of our approximate formulas increased as $\sigma_{\dB}$ decreased;
however, the approximate formulas where highly accurate when $\sigma_{\dB} \ge 3$. As discussed in Section~\ref{subsubsec-accura}, 
this result can be considered as realistic settings as a typical value of $\sigma_{\dB}$ lies between 3 and 15 dB~\cite{Gold05}. 
%
\begin{figure}[t]
\centering
\includegraphics[width=2.5in]{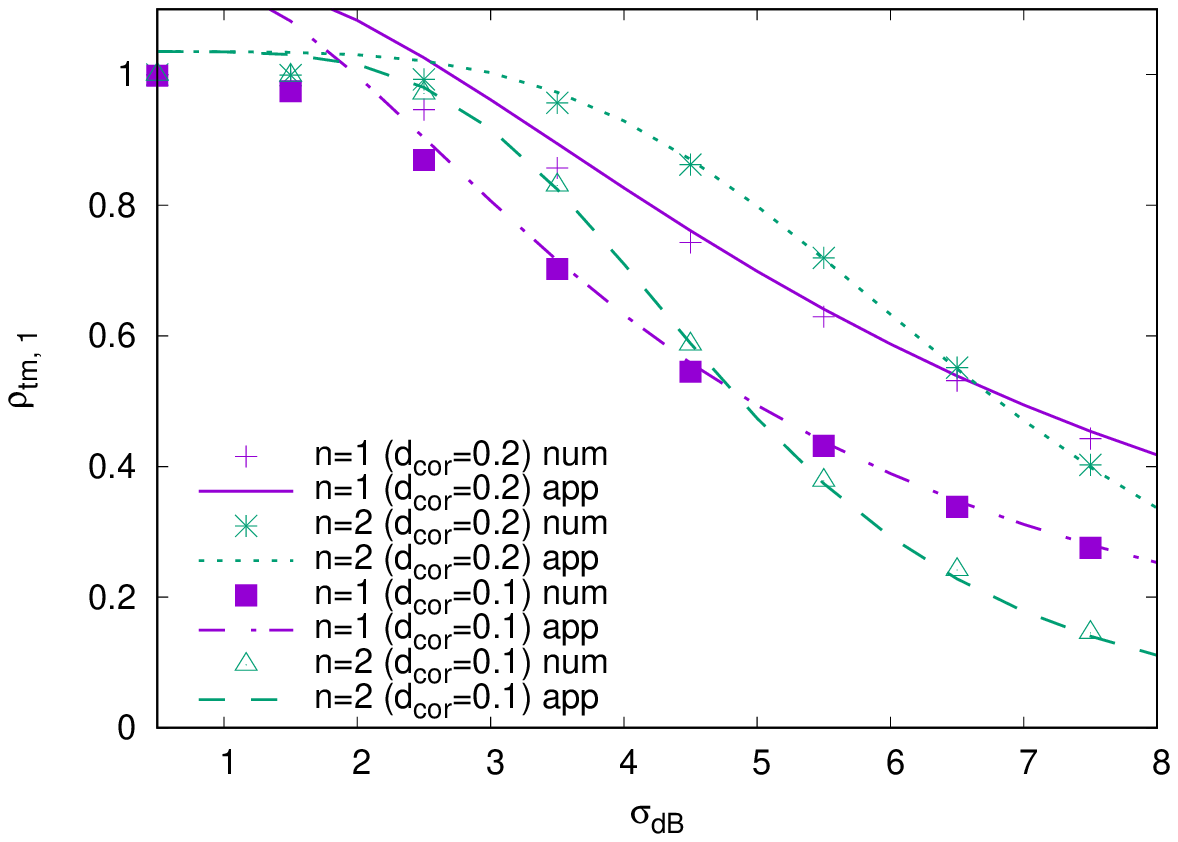}
\caption{Comparison of $\rho_{\rtm,1}$ from numerical integration and approximation under CIM model with different $d_{\mathrm{cor}}$ when varying $\sigma_{\dB}$. $v_{\max} = 0.05$.}
\label{fig:rho-tmp-sigma}
\vskip -7pt
\end{figure}
\section{Conclusion}
\label{sec-conclude}
In this paper, we studied the spatial and temporal correlation of interference in a correlated shadowing environment. On the basis of Gudmundson's model~\cite{Gudm91}, we derived a simple asymptotic expansion of the spatial and temporal correlation coefficients of interference when the variance of shadowing is large. The obtained expansions are helpful for understanding the relationship between the correlation distance of shadowing and correlation of interference. We also showed in numerical results that the asymptotic expansions can be used as tight approximate formulas, which is useful for modeling various wireless networks in spatially correlated shadowing environments. 

Furthermore, we can adopt our results for the design of MANETs, such as the correlation-aware retransmission scheme proposed by Gong and Haenggi~\cite{Gong14}. In this scheme, if a transmission fails in a temporally correlated interference environment, the retransmission is prolonged to prevent failure in consecutive time slots. On the other hand, if a transmission is successful, transmissions
should be made more frequently because the success probability in the following time-slots conditioned on the current success is expected to be higher. We can choose an appropriate retransmission period based on the results of the temporal correlation. Similarly, we can also consider a design problem with multiple antenna systems so that the received interferences become independent  based on the analytical results of the spatial correlation.  

In this work, we set several assumptions to consider a general setting and obtain tractable results. Thus, 
we will endeavor extensions to more realistic settings in our future work. 
For example, we assumed slotted ALOHA as the MAC.  
A CSMA network can be spatially modeled by using a hardcore point process 
(e.g., \cite{Tong16}). Such an extension to other spatial models will be considered in future studies. 
Furthermore, even though our results of the temporal correlation 
(Section~\ref{sec-temp}) do not require any conditions on the 
distribution of the moving distance,  i.e., $\psi(\cdot)$, 
we assumed a simple i.i.d. mobility model. 
On the other hand, the correlation of node mobility, such as the random waypoint model considered in \cite{Kouf18}, 
makes the analysis highly intractable even in a one-dimensional network. 
Thus, we will consider an extension to such a model as in future research.
In addition, we adopted Gudmundson's model~\cite{Gudm91} 
for modeling the spatially correlated shadowing; 
however, many studies have reported cases where this model is not applicable.
For example, Agrawal and Patwari \cite{Agra09} demonstarated such a case 
and proposed an extension of Gudmundson's model. 
However, their model includes many non-analytical integral expressions, 
and thus, tractable results cannot be obtained. Applying our approximation
approach to such a complicated model will also be attempted in a future work. 





\ifCLASSOPTIONcompsoc
  \section*{Acknowledgments}
\else
  \section*{Acknowledgment}
\fi

This work was supported by JSPS KAKENHI Grant Number 19K14981.

\ifCLASSOPTIONcaptionsoff
  \newpage
\fi

 \makeatletter
    \renewcommand{\theequation}{%
    \thesection.\arabic{equation}}
    \@addtoreset{equation}{section}
    \markboth{}{}

    \pagestyle{headings}
  \makeatother
 \setcounter{page}{1}

\appendices

\section{Proof of Corollary~\ref{lem-watson}}\label{appen-proof-of-lem-watson}
To begin, we give the following equation by reduction, 
for any $m \in \bbN$, 
\begin{align}
&
\varphi^{(m)}(s)= 
(u'(s))^{m+1}
g^{(m)}(a + u(s))
\nonumber
\\
&+ 
{m(m+1) \over 2}
(u'(s))^{m- 1}
u''(s)
g^{(m-1)}(a + u(s))
+ \zeta_{m-2}(s),
\label{eq-psi^m}
\end{align}
%
%
%
where $\zeta_m(s)$ denotes 
a certain polynomial function of 
$\{g^{(i)}(s); i \le m\}$. 
If $m = 1$, we can immediately
confirm that (\ref{eq-psi^m}) holds
because 
\begin{align}
\varphi'(s)
&= (u'(s))^2 g'(a + u(s)) +  u''(s)g(a + u(s)).
\nonumber
\end{align}
Furthermore, differentiating the 
both sides of (\ref{eq-psi^m}) with respect to $s$, 
we have
\begin{align}
&
\varphi^{(m+1)}(s)
= 
(u'(s))^{m+2}
g^{(m+1)}(a + u(s))
\nonumber
\\
&
~~+
(m + 1)
(u'(s))^{m}
u''(s)
g^{(m)}(a + u(s))
\nonumber
\\
&
~~+ 
{m(m+1) \over 2}
\left[
(u'(s))^{m}
u''(s)
g^{(m)}(a + u(s))
\right.
\nonumber
\\
&
~~+
\left.\{
(m -1)
(u'(s))^{m- 2}
(u''(s))^2
\right.
\nonumber
\\
&~~+
\left.
(u'(s))^{m- 1}
u'''(s)\}
g^{(m-1)}(a + u(s))
\right]
+ \zeta'_{m-2}(s)
\nonumber
\\
&=(u'(s))^{m+2}
g^{(m+1)}(a + u(s))
\nonumber
\\
&~~+
{(m+1)(m+2) \over 2}
(u'(s))^{m}
u''(s)
g^{(m)}(a + u(s))
+
\zeta_{m-1}(s),
\nonumber
\end{align}
where we use
\begin{align}
\zeta_{m-1}(s)&:= 
{m(m+1) \over 2}
\left.\{
(m -1)
(u'(s))^{m- 2}
(u''(s))^2
\right.
\nonumber
\\
&
\left.
+
(u'(s))^{m- 1}
u'''(s)\right\}
g^{(m-1)}(a + u(s))
+
\zeta'_{m-2}(s),
\nonumber
\end{align}
in the second equality. 
Therefore, (\ref{eq-psi^m}) holds for any $m \in \bbN$. 
%

By substituting $s = 0$ into (\ref{eq-psi^m}) 
and using (\ref{eq-g-cond-1}) and (\ref{eq-g-cond-2}), 
we obtain
%
\begin{align}
 \varphi^{(n-1)}(0) &= 
(u'(0))^{n} g^{(n-1)}(a),
\label{eq-psi-n-1}
\\
 \varphi^{(n)}(0) &= 
(u'(0))^{n+1} g^{(n)}(a)
\nonumber
\\
&
~~~~
+
{n (n + 1) \over 2}
(u'(0))^{n- 1}
u''(0) g^{(n-1)}(a).
\label{eq-psi-n}
\end{align}
%
%
%
In addition, by differentiating (\ref{eq-def-u}) with respect to $s$, we have
\[
R'(t) u'(s) = -1, \quad R''(t) u'(s) + R'(t) u''(s) = 0,
\]
which leads to 
\begin{align}
u'(0) &= - {1 \over R'(a)},\qquad
u''(0) = - {R''(a) \over (R'(a))^3}.
\nonumber
\end{align}
Finally, by substituting the above expressions into (\ref{eq-psi-n-1})
and (\ref{eq-psi-n}), and 
applying the result to (\ref{eq-watson-1}), we obtain 
(\ref{eq-watson-n>2}). 

\section{Proof of Proposition~\ref{lem-appen-asym}}\label{appen-proof-of-lem-appen-asym}
In this proof, we only prove the case of $n \ge 2$ because 
the case of $n=1$ can be easily shown in a similar manner as 
the proof of Lemma~3 in \cite{Kimu18}. 

In the  $n$-dimensional spherical coordinate system, 
by using $r \in \bbR_+$ and 
$\vc{\theta} = (\theta_1, \dots, \theta_{n-1})
\in [0, \pi)^{n-2}\times[0,2\pi) \triangleq \mathcal{P}$, 
a fixed point $\vc{x} = (x_1,x_2,\dots, x_n)$ in 
$n$-dimensional Cartesian coordinates is projected into
\begin{align}
x_1 &= r \cos\theta_1,
\nonumber
\\
x_2 &= r \sin\theta_1\cos\theta_2,
\nonumber
\\
x_3 &= r \sin\theta_1\sin\theta_2\cos\theta_3,
\nonumber
\\
\vdots
&\nonumber
\\
x_{n-1} &= r \sin\theta_1\cdots\sin\theta_{n-2}\cos\theta_{n-1},
\nonumber
\\
x_{n} &= r \sin\theta_1\cdots\sin\theta_{n-2}\sin\theta_{n-1}.
\end{align}
%
In addition, the Jacobian is given by
\[
r^{n-1}\sin^{n-2}\theta_{1}\cdots\sin\theta_{n-2}
\triangleq 
r^{n-1}J(\vc{\theta}).
\]
Thus, by considering $\vc{x} \to (r, \vc{\theta})$ and
$\vc{y} \to (s, \vc{\phi})$  $(r, s \in \bbR_+, \vc{\theta}, \vc{\phi}\in\calP)$,
$\Psi_n(f)$ can be rewritten as
\begin{align}
\Psi_n(f) &= 
\int_{\mathcal{P}}\int_{\mathcal{P}}
\int_{\bbR_+}\int_{\bbR_+}\!\!\!
{s^{n-1}r^{n-1} e^{\sigma^2 e^{ - {s \over d_0}}} f(s)
\over \ell(r)
\ell(D(r,s, \vc{\theta}, \vc{\phi}))}
\nonumber
\\
&~~~\times
J(\vc{\theta})
J(\vc{\phi})
\rmd s \rmd r
\rmd \vc{\theta} \rmd \vc{\phi},
%
\label{eq-Psi_n-sph}
\end{align}
where $D(r, s, \vc{\theta}, \vc{\phi})$ denotes the distance between two 
distinct points in the $n$-dimensional spherical coordinate system such that
\begin{align}
&D^2(r, s, \vc{\theta}, \vc{\phi}) = 
(r \sin\theta_1 \cdots  \sin\theta_{n-1}
-
s \sin\phi_1 \cdots \sin\phi_{n-1})^2
\nonumber
\\
&+
\sum_{i=1}^{n-1} 
(r \sin\theta_1 \cdots \sin\theta_{i-1} \cos\theta_{i}
-
s \sin\phi_1 \cdots \sin\phi_{i-1} \cos\phi_{i})^2
\nonumber
\\
&=
r^2 + s^2 - 2 rs \left(\prod_{i=1}^{n-1} \sin\theta_i\sin\phi_i
\right.
\nonumber
\\
&
\left.
\quad
+ 
\sum_{i=1}^{n-1}\cos \theta_i \cos\phi_i \prod_{j=1}^{i-1}\sin\theta_j\sin\phi_j
\right). 
\label{eq-def-D}
\end{align}
Let 
\[
g(s) = {s^{n-1} f(s) \over \ell(D(r, s, \vc{\theta}, \vc{\phi}))}.
\]
It then follows from (\ref{eq-def-D}) that 
\begin{align}
g'(0) &= g''(0) = \cdots = g^{(n-2)}(0) = 0, 
\label{eq-g-0}
\\
g^{(n-1)}(0) &= { (n - 1) ! f(0) \over \ell(r) },
\label{eq-g-n-1}
\\
g^{(n)}(0) 
&= n! \left(
{f'(0) \over \ell(r)}
+
\alpha r^{\alpha - 2} 
{f(0) \over (\ell(r))^2}
\left(\prod_{i=1}^{n-1} \sin\theta_i\sin\phi_i
\right.
\right.
\nonumber
\\
&
\left.
\left.
\quad
+ 
\sum_{i=1}^{n-1}\cos \theta_i \cos\phi_i \prod_{j=1}^{i-1}\sin\theta_j\sin\phi_j
\right)\right). 
\label{eq-g-n}
\end{align}
In addition, $(\rmd / \rmd s) e^{- {s \over d_0}}\mid_{s = 0} = - 1 / d_0$. 
Therefore, we can apply Corollary~\ref{lem-watson}
and obtain
\begin{align}
&\int_{\bbR_+}\!\!\!
{s^{n-1}  e^{\sigma^2 e^{- {s \over d_0}}} f(s) \over \ell(D(r, s, \vc{\theta}, \vc{\phi}))}
\rmd s 
\sim
d_0^n  e^{\sigma^2}
\left(
{g^{(n-1)}(0) \over \sigma^{2 n}}
+
\left(
d_0 g^{(n)}(0) 
\right.
\right.
\nonumber
\\
&\left.
\left.
+
{n (n +1)  g^{(n-1)}(0)\over 2} 
\right)
{1 \over \sigma^{2 (n+1)}}
+ O\left({1 \over \sigma^{2 (n+2)}}\right)
\right).
\label{eq-sim-Psi_temp}
\end{align}
Note that 
\begin{align}
&\int_{\calP}\int_{\calP}
\left(\prod_{i=1}^{n-1} \sin\theta_i\sin\phi_i
\right.
\nonumber
\\
&\qquad\left.+
\sum_{i=1}^{n-1}\cos \theta_i \cos\phi_i \prod_{j=1}^{i-1}\sin\theta_j\sin\phi_j\right)
\rmd \vc{\theta}\rmd \vc{\phi}
= 0.
\nonumber
\end{align}
Thus, combining the above expression with (\ref{eq-g-n-1}) and (\ref{eq-g-n}) and substituting
them into (\ref{eq-sim-Psi_temp}), we have
%
%
%
\begin{align}
&\Psi_n(f) \sim d_0^n e^{\sigma^2}
\int_{\calP}\int_{\calP}
\int_{\bbR_+}
{r^{n-1} \over (\ell(r))^2}
\left\{
{(n-1)!f(0)  \over \sigma^{2n}}
\right.
\nonumber
\\
&~~~~+ 
\left(
{(n + 1) ! f(0) \over 2}
+
d_0 n! f'(0) 
\right)
{1 \over \sigma^{2(n+1)}} 
\nonumber
\\
&~~~~\times
\left.
J(\vc{\theta})
J(\vc{\phi})
\rmd s \rmd r
\rmd \vc{\theta} \rmd \vc{\phi}
+ O\left({1 \over \sigma^{2 (n+2)}}\right)\right\}
\nonumber
\\
&=
d_0^{n}e^{\sigma^2} 
\int_{\bbR^n} {\rmd x \over (\ell(\|x\|))^2}
\int_{\calP}
J(\vc{\phi}) \rmd \vc{\phi}
\left\{
{(n-1)!f(0)  \over \sigma^{2n}}
\right.
\nonumber
\\
&~~~~
+
\left(
{(n + 1) ! f(0) \over 2}
+
d_0 n! f'(0) 
\right)
{1 \over \sigma^{2(n+1)}} 
\nonumber
\\
&
~~~
\left.{}
+O\left({1 \over \sigma^{2 (n+2)}}\right)
\right\}.
\nonumber
\end{align}
Finally, combining the above expression with (\ref{eq-E[I^2]-sec-term}) and
\[
\int_{\calP} J(\vc{\phi}) \rmd \vc{\phi} = 
{ \int_0^1 r^{n-1} \rmd r \int_{\calP}J(\vc{\phi}) \rmd \vc{\phi} \over  \int_0^1 r^{n-1} \rmd r } 
= n V_n,
\]
we obtain (\ref{eq-int-I^2-sub-last-n=2}). 

\section{Proof of Lemma~\ref{thm-corr-spatial}}
\label{proof-corr-spatial}


In this appendix, we prove the case of $n=1$ because the case of $n=2$ is already shown. 
From Campbell's theorem and (\ref{eq-assu-1}), 
the first term in (\ref{eq-spa-corr-0}) can be rewritten as
\begin{eqnarray}
\lefteqn{
\bbE
\left[
\sum_{x_i, x_j \in \Phi_{\neq}^{(2)}}
{ h_{i} \tilde{h}_{j} \calS(x_i) \calS(x_j) \over \ell(\|x_i\|) \ell(\|x_j - \tilde{o}_\delta\|)}
\right] 
} &&
\nonumber
\\
&=&
\lambda^2 p^2
\int \!\!\!
\int_{(\bbR)^2}
{ \exp\left(\sigma^2 e^{-{\|x - y\| + \delta \over d_0}}\right)
\over \ell(\|x + \tilde{o}_\delta\|) \ell(\|y\|)}
\rmd x \rmd y.
\label{eq-int-spa-cor-1-appen}
\end{eqnarray}
By considering $|x- y| = s \in \bbR_+$, the integral term in (\ref{eq-int-spa-cor-1-appen}) can be rewritten as
\begin{eqnarray}
\lefteqn{
\int_{\bbR}\int_{\bbR}
{e^{\sigma^2 e^{-{|x - y| + \delta \over d_0}}}
\over
\ell(\|x + \tilde{o}_\delta\|)
\ell(\|y\|)} \rmd y \rmd x
} &&
\nonumber
\\
&=&
\int_{\bbR}\!
\int_{\bbR_+} \!
{e^{\sigma^2 e^{-{s+\delta \over d_0}}} \over \ell(|x + \delta|)}
\left[
{1
\over
\ell(|x + s|)}
+
{1 \over
\ell(|x - s|)} \right] \rmd s \rmd x.
\nonumber
\\
&=&
\int_{\bbR}\!
{1 \over \ell(|x + \delta|)}
\int_{\bbR_+} \!
e^{\sigma^2 e^{-{s+\delta \over d_0}}}
\left[
{1
\over
\ell(|x| + s)}
+
{\dd{1}(s < |x|) \over
\ell(|x| - s)} 
\right.
\nonumber
\\
&{}&
+
\left.
{\dd{1}(s \ge |x|) \over
\ell(s - |x|)} 
\right] \rmd s \rmd x.
\label{eq-int-delta-main}
\nonumber
\end{eqnarray}
%
Applying Corollary~\ref{lem-watson}, we obtain, for $x \ge 0$, 
\begin{eqnarray}
&&\int_0^{\infty} 
{e^{\sigma^2 e^{-{s+\delta \over d_0}}}
\over
\ell(x + s)}
\rmd s
\sim
e^{\sigma^2 e^{- {\delta \over d_0}}}
\left[
{d_0 \over \eps_0 + x^\alpha} {e^{\delta \over d_0} \over \sigma^2}
\right.
\nonumber
\\
&&
\left.
+
\left(
{d_0 \over \eps_0 + x^\alpha} 
-
{\alpha d_0^2 x^{\alpha -1}\over (\eps_0 + x^{\alpha})^2} 
\right)
{e^{\delta \over d_0} \over \sigma^4}
+ O\left({e^{3 \delta \over d_0} \over \sigma^6}\right)
\right],\quad
\label{eq-int-delta-main-add-1}
\\
&&\int_0^{x} 
{e^{\sigma^2 e^{-{s+\delta \over d_0}}}
\over
\ell(x - s)}
\rmd s
\sim
e^{\sigma^2 e^{- {\delta \over d_0}}}
\left[
{d_0 \over \eps_0 + x^\alpha} {e^{\delta \over d_0} \over \sigma^2}
\right.
\nonumber
\\
&&
\left.
+
\left(
{d_0 \over \eps_0 + x^\alpha} 
+
{\alpha d_0^2 x^{\alpha -1} \over (\eps_0 + x^{\alpha})^2} 
\right)
{e^{\delta \over d_0} \over \sigma^4}
+ O\left({e^{3 \delta \over d_0} \over \sigma^6}\right)
\right].
\label{eq-int-delta-main-add-2}
\end{eqnarray}
In the same way as the above, we can also have, for $x \ge 0$, 
\begin{eqnarray}
\int_x^{\infty} 
{e^{\sigma^2 e^{-{s+\delta \over d_0}}}
\over
\ell(s - x)}
\rmd s
&=&
\int_0^{\infty} 
{e^{\sigma^2 e^{-{x + s+\delta \over d_0}}}
\over
\ell(s)}
\rmd s
\nonumber
\\
&\sim&
{d_0 \over \eps_0} {e^{x + \delta \over d_0} \over \sigma^2}
+ 
O\left(
{e^{2(x + \delta) \over d_0} \over \sigma^4}
\right).
\label{eq-int-delta-main-add-3}
\end{eqnarray}
Therefore, substituting (\ref{eq-int-delta-main-add-1})--(\ref{eq-int-delta-main-add-3}) into (\ref{eq-int-delta-main}) leads to 
\begin{eqnarray}
\lefteqn{
\int_{\bbR}\int_{\bbR}
{e^{\sigma^2 e^{-{|x - y| + \delta \over d_0}}}
\over
\ell(\|x + \tilde{o}_\delta\|)
\ell(\|y\|)} \rmd y \rmd x
\sim
\int_{\bbR}{1 \over \ell(|x + \delta|)}
}
\nonumber
\\
&\times&
\left[
e^{\sigma^2 e^{- {\delta \over d_0}}}
{2 d_0 \over \eps_0 + |x|^\alpha}
\left\{
{e^{\delta \over d_0} \over \sigma^2}
+
{e^{2 \delta \over d_0} \over \sigma^4}
+
O\left(
{e^{3 \delta \over d_0} \over \sigma^6}
\right)
\right\}
\right.
\nonumber
\\
&+&
\left.
e^{\sigma^2 e^{- {|x| + \delta \over d_0}}}
\left\{
{d_0 \over \eps_0} {e^{|x| + \delta \over d_0} \over \sigma^2}
+
O\left(
{e^{2(|x| + \delta) \over d_0} \over \sigma^4}
\right)
\right\}
\right]\rmd x.
\label{eq-int-delta-main-x-2}
\end{eqnarray}
Recall here that (see the definition of $U_1(\delta)$, i.e., (\ref{eq-def-U})) 
\begin{equation}
U_1(\delta) = \int_{\bbR} {\rmd x \over \ell(|x|)\ell(|x - \delta|)} 
= \int_{\bbR} {\rmd x \over \ell(|x + \delta|)(\eps_0 + |x|^\alpha)} .
\label{eq-def-U-1}
\end{equation}
In addition, it follows from Corollary~\ref{lem-watson} that
\begin{eqnarray}
\lefteqn{
\int_{\bbR}{e^{\sigma^2 e^{- {x + \delta \over d_0}}} e^{|x| + \delta \over d_0} \rmd x
 \over \ell(|x + \delta|)} 
=
\int_{0}^\infty{e^{\sigma^2 e^{- {x + \delta \over d_0}}} e^{x + \delta \over d_0} \rmd x
\over \ell(x + \delta)} 
} &&
\nonumber
\\
&{}&
+
\int_{0}^\delta{e^{\sigma^2 e^{- {x + \delta \over d_0}}} e^{x + \delta \over d_0} \rmd x
\over \ell(\delta - x)} 
+ 
\int_{\delta}^\infty{e^{\sigma^2 e^{- {x + \delta \over d_0}}} e^{x + \delta \over d_0} \rmd x
\over \ell(x - \delta)} 
\nonumber
\\
&\sim&
2 e^{\sigma^2 e^{- \delta \over d_0}}
\left[
{d_0 \over (\eps_0 + \delta^{\alpha})} {e^{2\delta \over d_0} \over \sigma^2}
+ 
O\left({e^{4\delta \over d_0} \over \sigma^4} \right)
\right]
\nonumber
\\
&{}&
+
e^{\sigma^2 e^{- 2 \delta \over d_0}}
\left[
{d_0 \over \eps_0} {e^{4\delta \over d_0} \over \sigma^2}
+
O\left({e^{6\delta \over d_0} \over \sigma^4} \right)
\right].
\nonumber
\end{eqnarray}
Thus, combining the above and (\ref{eq-def-U-1}) with (\ref{eq-int-delta-main-x-2}) yields
\begin{eqnarray}
\lefteqn{
\int_{\bbR}\int_{\bbR}
{e^{\sigma^2 e^{-{|x - y| + \delta \over d_0}}}
\over
\ell(\|x + \tilde{o}_\delta\|)
\ell(\|y\|)} \rmd y \rmd x
}&&
\nonumber
\\
&\sim&
2 d_0 e^{\sigma^2 e^{ {-\delta \over d_0}}}
U_1(\delta)
\left[
{e^{\delta \over d_0} \over \sigma^2}
+
{e^{2\delta \over d_0} \over \sigma^4}
+
O\left({e^{3\delta \over d_0} \over \sigma^6}\right)
\right]
\nonumber
\\
&{}&
+
e^{\sigma^2 e^{ {-\delta \over d_0}}}
{ 2 d_0^2 e^{2 \delta \over d_0} \over (\eps_0 + \delta^\alpha) \eps_0\sigma^4}
+
e^{\sigma^2 e^{- {2\delta \over d_0}}}
{d_0^2 e^{4 \delta \over d_0} \over \eps_0^{2}\sigma^4}.
\label{eq-int-delta-main-x-3}
\end{eqnarray}
Furthermore, recall that the second term in (\ref{eq-spa-corr-0}) can be rewritten as
\begin{eqnarray}
\bbE \left[
\sum_{x_i\in \Phi} {h_i \tilde{h}_i \calS_i \over \ell(\|x_i\|)\ell(\|x_i - \tilde{o}_\delta\|) }\right]
=
\lambda p e^{\sigma^2 e^{- \delta \over d_0}} U_1(\delta).
\nonumber
\end{eqnarray}
Consequently, putting (\ref{eq-int-delta-main-x-3}) into (\ref{eq-int-spa-cor-1-appen}) and combining it and 
the above into (\ref{eq-spa-corr-0}) completes the proof. 
\end{document}